\documentclass[a4paper,twocolumn,11pt,accepted=2022-04-01]{quantumarticle}
\pdfoutput=1
\usepackage[utf8]{inputenc}
\usepackage[english]{babel}
\usepackage{amsmath,amsfonts,amssymb}
\usepackage{hyperref}
\usepackage{tikz}
\usepackage{lipsum}
\usepackage{booktabs}
\usepackage[numbers,sort&compress]{natbib}
\usepackage[ruled]{algorithm2e}
\usepackage{latexsym}
\usepackage{subfigure}
\usepackage{lineno}
\usepackage{setspace}
\newtheorem{theorem}{Theorem}
\newtheorem{definition}{Definition}
\newtheorem{lemma}{Lemma}
\def\qed{\hfill $\Box$}

\newcommand{\red}[1]{\textcolor{black}{#1}}
\newcommand{\state}{\varphi}

\newcommand*\patchAmsMathEnvironmentForLineno[1]{
  \expandafter\let\csname old#1\expandafter\endcsname\csname #1\endcsname
  \expandafter\let\csname oldend#1\expandafter\endcsname\csname end#1\endcsname
  \renewenvironment{#1}
     {\linenomath\csname old#1\endcsname}
     {\csname oldend#1\endcsname\endlinenomath}}
\newcommand*\patchBothAmsMathEnvironmentsForLineno[1]{
  \patchAmsMathEnvironmentForLineno{#1}
  \patchAmsMathEnvironmentForLineno{#1*}}
\AtBeginDocument{
\patchBothAmsMathEnvironmentsForLineno{equation}
\patchBothAmsMathEnvironmentsForLineno{align}
\patchBothAmsMathEnvironmentsForLineno{flalign}
\patchBothAmsMathEnvironmentsForLineno{alignat}
\patchBothAmsMathEnvironmentsForLineno{gather}
\patchBothAmsMathEnvironmentsForLineno{multline}
}

\makeatletter
\renewcommand{\@algocf@capt@plain}{above}
\makeatother

\makeatletter
\newenvironment{proof}[1][\proofname]{\par
  \normalfont
  \topsep6\p@\@plus6\p@ \trivlist
  \item[\hskip\labelsep{\bfseries #1}\@addpunct{\bfseries.}]\ignorespaces
}{
  \endtrivlist
}
\makeatother

\makeatletter
\newcommand{\settheoremtag}[1]{% \settheoremtag{<tag>}
  \let\oldthetheorem\thetheorem% Store \thetheorem
  \renewcommand{\thetheorem}{#1}% Redefine it to a fixed value
  \g@addto@macro\endtheorem{% At \end{theorem}, ...
    \addtocounter{theorem}{-1}% ...restore theorem counter value and...
    \global\let\thetheorem\oldthetheorem}% ...restore \thetheorem
  }
\makeatother

\begin{document}

\title{Computationally Efficient Quantum Expectation with Extended Bell Measurements}

\author{Ruho Kondo}
\affiliation{Toyota Central R\&D Labs., Inc., 41-1, Yokomichi, Nagakute, Aichi 480-1192, Japan}
\email{r-kondo@mosk.tytlabs.co.jp}

\author{Yuki Sato}
\affiliation{Toyota Central R\&D Labs., Inc., 41-1, Yokomichi, Nagakute, Aichi 480-1192, Japan}

\author{Satoshi Koide}
\affiliation{Toyota Central R\&D Labs., Inc., 41-1, Yokomichi, Nagakute, Aichi 480-1192, Japan}

\author{Seiji Kajita}
\affiliation{Toyota Central R\&D Labs., Inc., 41-1, Yokomichi, Nagakute, Aichi 480-1192, Japan}

\author{Hideki Takamatsu}
\affiliation{Toyota Motor Corporation, 1 Toyota-Cho, Toyota, Aichi 471-8571, Japan}

\maketitle

%\linenumbers

\begin{abstract}
Evaluating an expectation value of an arbitrary observable $A\in{\mathbb C}^{2^n\times 2^n}$ through na\"ive Pauli measurements requires a large number of terms to be evaluated.
We approach this issue using a method based on Bell measurement, which we refer to as the extended Bell measurement method.
This analytical method quickly assembles the $4^n$ matrix elements into at most $2^{n+1}$ groups for simultaneous measurements in $O(nd)$ time, where $d$ is the number of non-zero elements of $A$.
The number of groups is particularly small when $A$ is a band matrix.
When the bandwidth of $A$ is $k=O(n^c)$, the number of groups for simultaneous measurement reduces to $O(n^{c+1})$.
In addition, when non-zero elements densely fill the band, the variance is $O((n^{c+1}/2^n)\,{\rm tr}(A^2))$, which is small compared with the variances of existing methods.
The proposed method requires a few additional gates for each measurement, namely one Hadamard gate, one phase gate and at most $n-1$ CNOT gates.
Experimental results on an IBM-Q system show the computational efficiency and scalability of the proposed scheme, compared with existing state-of-the-art approaches.
Code is available at \url{https://github.com/ToyotaCRDL/extended-bell-measurements}.
\end{abstract}

%==================================================
\section{Introduction}
%==================================================
Variational quantum algorithms (VQAs)~\cite{vqa,vqa2} have attracted attention in recent years because they are suitable for so-called Noisy Intermediate-Scale Quantum (NISQ) computers~\cite{nisq}, which do not have a fault-tolerant mechanism and have only tens or hundreds of qubits.
The VQA was first introduced in the field of chemistry but it has wide application, such as in solving linear systems~\cite{linear,vqls,xu2021variational,poisson1,poisson2}, optimization~\cite{qaoa1,qaoa2,qaoa3}, and machine learning~\cite{qml1,qml2,qml3,qml4,qml5,qml6,qml7}.
In VQAs, the quantum state, $\big|\psi\big>$, is parametrized by learning parameters held on a classical computer, and is used to calculate the expectation value, $\big<\psi\big|A\big|\psi\big>$, whose $A$ depends on the task.
The desired quantum state is then obtained by minimizing the expectation value using a classical optimization algorithm.

The bottleneck of a VQA relates to the computational cost of calculating the expectation value~\cite{towards}.
A popular method of evaluating the expectation value is a Pauli measurement that writes $A$ in the Pauli basis; i.e., $A=\sum^\mathcal{N}_{i=1} c_i {\bf P}_i$, where ${\bf P}_i$ is the tensor product of Pauli matrices and $c_i={\rm Tr}\big(A^\dagger {\bf P}_i\big)/2^n$ is the corresponding coefficient, with $(\bullet)^\dagger$ indicating the Hermitian conjugate.
The drawback of a na\"ive Pauli measurement is that it requires a large number of unique circuits to evaluate $\big<\psi\big|A\big|\psi\big>=\sum^\mathcal{N}_{i=1}c_i\big<\psi\big|{\bf P}_i\big|\psi\big>$ where $\mathcal{N}$ reaches $4^n$ when $A$ is a fully dense matrix.

Reducing the number of unique circuits for Pauli measurements is an interesting topic relating to the VQA~\cite{coloring,qwc,unitarypartitioning1,unitarypartitioning2,gc,sortedinsertion,entanglemeas}.
Both qubit-wise commuting (QWC)~\cite{qwc0,efficient,qwc} and general commuting (GC)~\cite{gc} methods, and their variants~\cite{unitarypartitioning1,unitarypartitioning2,entanglemeas}, reduce the number of unique circuits on the basis that commutative operators can be simultaneously diagonalized with the same unitary; i.e., commutative Pauli bases are simultaneously measurable.
This approach successfully reduces the number of unique circuits but its weakness is that finding the smallest groups for simultaneous measurements itself is an NP-hard problem.
These methods approximate solutions that are unfortunately not optimal in general.
In other words, a smaller number of unique circuits results in a longer calculation time, and a shorter calculation time results in a larger number of unique circuits.

Another problem is that the expression of general matrices in the Pauli basis is numerically inefficient. 
Suppose a sparse matrix $A$ that has a few non-zero elements.
In the extreme case, for example, if there exists just one non-zero element in $A$, then $2^n$ Pauli strings are required to represent $A$.
Despite the inefficiency of the Pauli basis, the sparse matrix is widely used in the domain of computer science.
A typical application is the finite element method (FEM)~\cite{fem} used for structural analysis,  fluid dynamics, heat transfer, electromagnetic analysis and many other continuum mechanics.
In the FEM, three dimensional real-space is discretized by a finite number of elements and the stiffness of the system is constructed by the nodal coordinates and properties of each element.
The resulting global stiffness matrix, $A$, in an appropriate node number ordering represents a band matrix because only the neighboring nodes in the elements can interact with each other.
Then, solving $Ax=b$, where $b$ denotes the external inputs, gives the desired response $x$.
Solving $Ax=b$ using VQAs has long been studied~\cite{linear,vqls,xu2021variational,poisson1,poisson2} but most of the VQAs limit the form of $A$.
For example, the assumption is that $A$ can be represented as the sum of a few weighted unitaries or weighted Pauli strings~\cite{linear,vqls,xu2021variational} or can be decomposed into a few measurable terms using the problem--specific conditions~\cite{poisson1,poisson2}.
Recently, an efficient expectation estimation method for a sparse Hamiltonian has been developed~\cite{kirby2021variational}, but it remains unsuitable for the NISQ device because of its \red{two-qubit entangling gate} counts.
Efficient expectation evaluations for the non-Pauli observables in terms of both the number of measurements and the circuit complexity are in great demand, such as for FEM applications, but the above issues have not yet been resolved.

A completely different approach that can be taken to reduce the computational cost for the expectation evaluation is to use the so-called classical shadow~\cite{shadow,shadow2,shadow3}, which reduces the number of measurements for the Pauli observables.
However, in the case of non-Pauli observables, this approach is not suitable for NISQ devices because it requires $O(n^2/\log(n))$ \red{two-qubit entangling gates} as pointed out in the original literature~\cite{shadow}.

In this paper, we propose an alternative method that we refer to as the \red{\emph{extended Bell measurement (XBM)}}.
This method is based on not Pauli measurements but Bell measurements.
The basic idea is inspired by Ref.~\cite{poisson1}, where the expectation value required to solve the Poisson equation was evaluated using Bell measurements.
We generalize this method for an arbitrary expectation evaluation.
In addition, we derive that the present method is closely related to the classical shadow with random Clifford measurements~\cite{shadow}.
The XBM method has the following advantages over existing methods.
%==============================
\begin{enumerate}
\item Fast grouping: The XBM method does not solve any difficult NP-hard problems, but it analytically groups the simultaneous measurements. 
\item Small number of circuits: The method requires the number of unique
circuits to be of the same order as the empirical number of an existing state-of-the-art
Pauli measurement-based method, such as the GC method.
\item Fewer \red{two-qubit entangling gates}: The method takes at most $n-1$ CNOT gates for each measurement operator.
\item Tight upper bound of the variance when $A$ is the densely filled band matrix.
\end{enumerate}
%==============================
Concretely, when $A$ is a band matrix whose bandwidth is $k=O(n^c)$, the number of distinct measurement operators is $O(n^{c+1})$.
In addition, when the non-zero elements are densely filled in the band, the upper bound of the variance of the expectation evaluation is $O((n^{c+1}/2^n)\,{\rm tr}(A^2))$, which is smaller than that of the classical shadow~\cite{shadow}.

Our algorithm is implemented with Qiskit~\cite{qiskit} and its code is available at \url{https://github.com/ToyotaCRDL/extended-bell-measurements}.
Additionally, real-device experiments are carried out using an IBM-Q system.
The results show that the accuracy of the expectation evaluation is comparable to that of na\"ive Pauli measurements and QWC but the execution time is shorter than the execution times of the existing methods when $A$ is a band matrix.

The remainder of the paper is organized as follows:
Section \ref{sec:method} describes the XBM method proposed in this paper.
Section \ref{sec:algorithm} presents algorithms of the XBM method and derives the computational complexity.
Section \ref{sec:experiment} presents the experimental results.
Section \ref{sec:conclusions} concludes the paper.

%==================================================
\section{Method}
\label{sec:method}
%==================================================
\red{
%==================================================
\subsection{Two-qubit Example of Proposed Method}
%==================================================
Before describing the general proposed method, we show an example with two-qubit.
Let $|\varphi\rangle \in {\mathbb C}^{4}$ be an arbitrary two-qubit quantum state.
Consider the following expectation value:
\[\displaystyle{
\langle \varphi| A | \varphi \rangle
= \langle \varphi|
\begin{bmatrix}
A_{00} & A_{01} & A_{02} & A_{03} \\
A_{10} & A_{11} & A_{12} & A_{13} \\
A_{20} & A_{21} & A_{22} & A_{23} \\
A_{30} & A_{31} & A_{32} & A_{33} \\
\end{bmatrix}
|\varphi \rangle
}\]
\[\displaystyle{
%=\sum^3_{b=0} A_{bb} \langle \varphi| b\rangle \langle b  |\varphi \rangle
%+\sum^3_{b=0}\sum^3_{c=0,c\neq b} A_{bc} \langle \varphi| b \rangle \langle c  |\varphi \rangle
=\sum^3_{b=0}\sum^3_{c=0} A_{bc} \langle \varphi| b \rangle \langle c  |\varphi \rangle.
}\]
The diagonal components of $A$ can be easily evaluated as
\[\displaystyle{
\sum^3_{b=0}A_{bb}\langle \varphi|b\rangle\langle b  |\varphi \rangle
=\sum^3_{b=0}A_{bb}|\langle b | \varphi \rangle|^2
}\]
where $|\langle b | \varphi \rangle|^2$ is the probability of obtaining a bit-string $b$ when measuring the quantum state $|\varphi\rangle$ in the computational basis.
Meanwhile, evaluating the off-diagonal parts of $A$ is not trivial.
Consider $A_{12}$ for example.
The expectation value of $A_{12}$ can be rearranged as
\[\displaystyle{
\langle \varphi|
\begin{bmatrix}
0 & 0 & 0 & 0 \\
0 & 0 & A_{12} & 0 \\
0 & 0 & 0 & 0 \\
0 & 0 & 0 & 0 \\
\end{bmatrix}
|\varphi \rangle
=\langle \varphi|\Big(A_{12}|1\rangle \langle 2|\Big)|\varphi \rangle
}\]
\[\displaystyle{
=A_{12}\Big(
{\rm Re}\big(\langle \varphi | 1 \rangle \langle 2 | \varphi \rangle\big)
+{\rm i}{\rm Im}\big(\langle \varphi | 1 \rangle \langle 2 | \varphi \rangle\big)
\Big)
}\]
where ${\rm Re}(\bullet)$ and ${\rm Im}(\bullet)$ are the real and imaginary parts of $(\bullet)$, respectively.
The real part of $\langle \varphi | 1 \rangle \langle 2 | \varphi \rangle$ is further rearranged as 
\[\displaystyle{
{\rm Re}\big(\langle \varphi | 1 \rangle \langle 2 | \varphi \rangle\big)
}\]
\[\displaystyle{
=\frac{1}{2}\Bigg(
\Bigg|\frac{\langle 01 | + \langle 10 | }{\sqrt{2}} | \varphi \rangle\Bigg|^2
-\Bigg|\frac{\langle 01 | - \langle 10 | }{\sqrt{2}} | \varphi \rangle\Bigg|^2
\Bigg)
}\]
where $|01\rangle=|1\rangle$ and $|10\rangle=|2\rangle$, respectively.
The quantum states $(|01\rangle \pm |10\rangle)/\sqrt{2}$ can be prepared using the Hadamard gate and CNOT gate as
\[\displaystyle{
\frac{| 01 \rangle + | 10 \rangle }{\sqrt{2}}
={\rm CNOT}(1,0){\rm H}(1) |01\rangle
}\]
\[\displaystyle{
\frac{| 01 \rangle - | 10 \rangle }{\sqrt{2}}
={\rm CNOT}(1,0){\rm H}(1) |11\rangle
}\]
where ${\rm H}(j)$ denotes the Hadamard gate acting on $j$th qubit, and ${\rm CNOT}(j,k)$ is the CNOT gate whose control and target qubits are $j$th and $k$th qubits, respectively.
Here, the rightmost bit is 0th bit whereas the leftmost bit is $(n-1)$th bit.
Using these preparation unitaries, ${\rm Re}\big(\langle \varphi | 1 \rangle \langle 2 | \varphi \rangle\big)$ is written as
\[\displaystyle{
{\rm Re}\big(\langle \varphi | 1 \rangle \langle 2 | \varphi \rangle\big)
=\frac{1}{2}\Bigg(
\Big|\langle 01 | {\rm H}(1){\rm CNOT}(1,0) | \varphi \rangle\Big|^2
}\]
\[\displaystyle{
-\Big|\langle 11 | {\rm H}(1){\rm CNOT}(1,0) | \varphi \rangle\Big|^2
\Bigg).
}\]
This value can be evaluated by the probabilities of obtaining bit-strings $01$ and $11$ when measuring ${\rm H}(1){\rm CNOT}(1,0) | \varphi\rangle$ in the computational basis.
In the same way,  ${\rm Im}\big(\langle \varphi | 1 \rangle \langle 2 | \varphi \rangle\big)$ is written as
\[\displaystyle{
{\rm Im}\big(\langle \varphi | 1 \rangle \langle 2 | \varphi \rangle\big)
=\frac{1}{2}\Bigg(
\Big|\langle 01 | {\rm H}(1){\rm CNOT}(1,0){\rm S}^\dagger(1) | \varphi \rangle\Big|^2
}\]
\[\displaystyle{
-\Big|\langle 11 | {\rm H}(1){\rm CNOT}(1,0){\rm S}^\dagger(1) | \varphi \rangle\Big|^2
\Bigg)
}\]
where ${\rm S}(j)$ denotes the phase gate acting on $j$th qubit.
This value can be evaluated by the probabilities of obtaining bit-strings $01$ and $11$ when measuring ${\rm H}(1){\rm CNOT}(1,0){\rm S}^\dagger(1) | \varphi\rangle$ in the computational basis.
In summary,
\[\displaystyle{
\langle \varphi|
\begin{bmatrix}
0 & 0 & 0 & 0 \\
0 & 0 & A_{12} & 0 \\
0 & 0 & 0 & 0 \\
0 & 0 & 0 & 0 \\
\end{bmatrix}
|\varphi \rangle
}\]
\[\displaystyle{
=A_{12}\Bigg[
\frac{1}{2}\Bigg(
\Big|\langle 01 | M^{(1,2)}_{\rm Re} | \varphi \rangle\Big|^2
-\Big|\langle 11 | M^{(1,2)}_{\rm Re} | \varphi \rangle\Big|^2
\Bigg)
}\]
\[\displaystyle{
+\frac{{\rm i}}{2}\Bigg(
\Big|\langle 01 | M^{(1,2)}_{\rm Im} | \varphi \rangle\Big|^2
-\Big|\langle 11 | M^{(1,2)}_{\rm Im} | \varphi \rangle\Big|^2
\Bigg)
\Bigg]
}\]
where
\[\displaystyle{
M^{(1,2)}_{\rm Re} := {\rm H}(1){\rm CNOT}(1,0)
}\]
and
\[\displaystyle{
M^{(1,2)}_{\rm Im} := {\rm H}(1){\rm CNOT}(1,0){\rm S}^\dagger(1),
}\]
respectively.
The remaining parts of $A$ can be evaluated in the same manner.
}

\red{
The most important part of the proposed method is the simultaneous measurements.
That is, some components of $A$ can be evaluated using the same unitary.
For example, $A_{03}$ can be evaluated using the same unitaries as those for $A_{12}$, ${\rm H}(1){\rm CNOT}(1,0)$ and ${\rm H}(1){\rm CNOT}(1,0){\rm S}^\dagger(1)$, because
{
\setlength{\abovedisplayskip}{0pt}
\setlength{\belowdisplayskip}{0pt}
\[\displaystyle{
(|0\rangle+|3\rangle)/\sqrt{2}
= {\rm CNOT}(1,0) {\rm H}(1) |00\rangle,
}\]
\[\displaystyle{
(|0\rangle-|3\rangle)/\sqrt{2}
= {\rm CNOT}(1,0) {\rm H}(1) |10\rangle,
}\]
\[\displaystyle{
(|0\rangle+{\rm i}|3\rangle)/\sqrt{2}
= {\rm S}(1){\rm CNOT}(1,0) {\rm H}(1) |00\rangle,
}\]
\[\displaystyle{
(|0\rangle-{\rm i}|3\rangle)/\sqrt{2}
= {\rm S}(1){\rm CNOT}(1,0) {\rm H}(1) |10\rangle.
}\]
 In addition, $A_{21}$ and $A_{30}$ can be evaluated using the same unitaries as those for $A_{12}$ and $A_{03}$, respectively, because ${\rm Re}\big(\langle \varphi | b \rangle \langle c | \varphi \rangle \big) = {\rm Re}\big(\langle \varphi | c \rangle \langle b | \varphi \rangle \big)$ and ${\rm Im}\big(\langle \varphi | b \rangle \langle c | \varphi \rangle \big) = - {\rm Im}\big(\langle \varphi | c \rangle \langle b | \varphi \rangle \big)$.
Similarly, the real and imaginary parts of $A_{01}$, $A_{10}$, $A_{23}$ and $A_{32}$ can be evaluated using ${\rm H}(0)$ and ${\rm H}(0){\rm S}^\dagger(0)$, respectively, whereas the real and imaginary of $A_{02}$, $A_{13}$, $A_{20}$ and $A_{31}$ can be evaluated using ${\rm H}(1)$ and ${\rm H}(1){\rm S}^\dagger(1)$, respectively.
}
}

%\the\abovedisplayskip{}

%\the\belowdisplayskip{}

%==================================================
\subsection{Expectation Estimation}
%==================================================

%==================================================
\renewcommand{\arraystretch}{1.5}
\begin{table*}[tb]
\centering
\caption{Definitions of the symbols appearing in Eq.~\eqref{eq:exp}.}
\label{tab:def}
\begin{tabular}{ll}
\toprule
Symbol & Definition\\
\midrule
%${\mathcal{S}(A)} $ & $\{b \oplus c | A_{bc}\neq 0, b \neq c\}$\\
${\mathcal{S}(A)} $ & $\{b \oplus c | A_{bc}\neq 0\}$\\
$\oplus$ & bit-wise exclusive OR\\
$(b',c',w)$ & $(b,b\oplus l,0)$ for $b < b\oplus l$\\
 & $(b\oplus 2^{j_0^{(l)}}, b\oplus l\oplus 2^{j_0^{(l)}},1)$ for $b > b \oplus l$\\
$j^{(l)}_0$ & $\max \{j | l_j = 1\}$\\
$l_j$ & a value of $j$th bit of $l$ in the binary expression\\
$A_{(bc)}$ & $(A_{bc}+A_{cb})/2$ \\
$A_{[bc]}$ & $(A_{bc}-A_{cb})/2$ \\
$M^{(l)}_{\rm Re}$ & ${\rm H}(j_0^{(l)}) \Big( \prod_{k\in \widetilde{\mathcal{T}}^{(l)}}{\rm CNOT}(j_0^{(l)},k)\Big) $ \\
$M^{(l)}_{\rm Im}$ & ${\rm H}(j^{(l)}_0) \Big( \prod_{k\in \widetilde{\mathcal{T}}^{(l)}}{\rm CNOT}(j_0^{(l)},k) \Big){\rm S}^\dagger(j^{(l)}_0)$ \\
$\widetilde{\mathcal{T}}^{(l)}$ &  $\{j|l_j=1\}\setminus\{j^{(l)}_0\}$\\
${\rm H}(j)$ & Hadamard gate acting on the $j$th qubit \\
${\rm S}(j)$ & phase gate acting on the $j$th qubit \\
${\rm CNOT}(j,k)$ & CNOT gate whose control and target bits are the $j$th and $k$th qubits\\
\bottomrule
\end{tabular}
\end{table*}
%==================================================

\red{
We now generalize the construction of $\langle \varphi | A | \varphi \rangle$ for arbitrary number of qubits.
}
Let $A\in{\mathbb C}^{2^n\times 2^n}$ be an arbitrary matrix and $|\state\rangle\in{\mathbb C}^{2^n}$ be an arbitrary $n$-qubit quantum state.
The goal is to express $\langle \state |A|\state \rangle$ in the form ${\mathbb{E}}[\hat{a}]=\sum_i a_i \, P(\hat{a}=a_i)$, where $\hat{a}$ is a complex random variable, $P$ is a probability distribution on the bit-strings obtained from a quantum circuit measurement, and $a_i$ is an outcome, which is calculated from a measured bit-string and the components of $A$.
The more general form, $\langle \psi_0 | A | \psi_1 \rangle$ (where $|\psi_0\rangle \neq |\psi_1\rangle$), can also be expressed in the same form (see Appendix~\ref{app:derivation}).

\red{Here, we show a brief derivation of the closed-form expression of $\langle \state |A|\state \rangle$. In the subsequent paragraphs, we give an intuitive description of the derivation (see Appendix~\ref{app:derivation} for a detailed derivation).}
In this paper, the matrix $A$ is expressed as
%==============================
\begin{align}
\label{eq:Adecomp}
A=\sum_{b,c}A_{bc}|b\rangle \langle c|
\end{align}
%==============================
where $A_{bc}$ is the $(b,c)$-component of $A$.
Using this representation, we give the closed-form expression of $\langle \varphi | A | \varphi \rangle$ derived in this study\red{, where the notations appearing in Eq.~\eqref{eq:exp} are defined in Table~\ref{tab:def}}:
%==============================
\begin{widetext}
\newcommand{\vs}{0mm}
\begin{align}
\label{eq:exp}
\langle \state |A|\state \rangle
&= \sum^{2^n-1}_{b=0} A_{bb} |\langle b | \state \rangle |^2  + \sum^{2^n-1}_{b=0}\sum^{2^n-1}_{c=0,c\neq b} A_{bc} \Big[ {\rm Re}\big( \langle \varphi | b \rangle \langle c | \varphi \big) + {\rm i} \, {\rm Im}\big( \langle \varphi | b \rangle \langle c | \varphi \big) \Big] \nonumber \\[\vs]
&= \sum^{2^n-1}_{b=0} A_{bb} |\langle b | \state \rangle |^2 
+ \sum^{2^n-1}_{b=0}\sum^{2^n-1}_{c=0, c\neq b} A_{bc} \Bigg[
\frac{1}{2} \Bigg(\Bigg|\frac{\langle b | + \langle c |}{\sqrt{2}} | \varphi \rangle \Bigg|^2 - \Bigg|\frac{\langle b | - \langle c |}{\sqrt{2}} | \varphi \rangle \Bigg|^2\Bigg) \nonumber \\[0mm]
&{\hspace{20mm}}
+\frac{{\rm i}}{2} \Bigg(\Bigg|\frac{\langle b | - {\rm i} \langle c |}{\sqrt{2}} | \varphi \rangle \Bigg|^2 - \Bigg|\frac{\langle b | + {\rm i} \langle c |}{\sqrt{2}} | \varphi \rangle \Bigg|^2\Bigg)
\Bigg]\nonumber \\[\vs]
&= \sum^{2^n-1}_{b=0} A_{bb} |\langle b | \state \rangle |^2 
+ \sum^{2^n-1}_{b=0}\sum^{2^n-1}_{c=0,c\neq b} A_{bc} \Bigg[
\frac{1}{2} \Big( \big|\langle b | M^{(b,c)}_{\rm Re} | \varphi \rangle \big|^2 - \big|\langle b \oplus 2^{j^{(b,c)}_0} | M^{(b,c)}_{\rm Re} | \varphi \rangle \big|^2 \Big) \nonumber \\[0mm]
&{\hspace{20mm}} +\frac{{\rm i}}{2} \Big( \big|\langle b | M^{(b,c)}_{\rm Im} | \varphi \rangle \big|^2 - \big|\langle b \oplus 2^{j^{(b,c)}_0} | M^{(b,c)}_{\rm Im} | \varphi \rangle \big|^2 \Big)
\Bigg]\nonumber \\[\vs]
&= \sum^{2^n-1}_{b=0} A_{bb} |\langle b | \state \rangle |^2 
+ \sum_{l\in \mathcal{S}(A)\red{\setminus\{0^n\}}}\sum^{2^{n}-1}_{b=0} \Bigg((-1)^w A_{(b'c')} |\langle b | M^{(l)}_{\rm Re}| \state \rangle |^2 + {\rm i}(-1)^w A_{[b'c']} |\langle b | M^{(l)}_{\rm Im}| \state \rangle |^2\Bigg)\nonumber\\[\vs]
&= \sum_{l\in\mathcal{S}(A)} \sum_{s\in\{{\rm Re},{\rm Im}\}} \sum^{2^n-1}_{b=0} a(l,s,b)\, p(l,s,b)
\end{align}
\end{widetext}
%==============================
where $a(l,s,b)$ and $p(l,s,b)$ are respectively defined as
%==============================
\begin{align}
\label{eq:coeff}
a(l,s,b):=\left\{
        \begin{array}{ll}
            \displaystyle{ \frac{A_{bb}}{p(l,s)} } & l=0^n \\[5mm]
            \displaystyle{ \frac{(-1)^w A_{(b'c')}}{p(l,s)} } & l\neq 0^n,\ s={\rm Re}  \\[5mm]
            \displaystyle{ \frac{{\rm i}(-1)^w A_{[b'c']}}{p(l,s)} } & l\neq 0^n, \ s={\rm Im}
        \end{array}
    \right. 
\end{align}
%==============================
and
%==============================
\begin{align}
p(l,s,b) := p(l,s)|\langle b | M^{(l)}_s | \state \rangle |^2.
\end{align}
%==============================
Note that $M^{(0^n)}_{s}=I$ and $p(l=0^n,s)=p(l=0^n)$ in this representation.
$p(l,s)$ is the probability of selecting the measurement operator $M^{(l)}_s$, which can be arbitrarily defined.
%The definitions of other symbols appearing in Eq.~\eqref{eq:exp} are summarized in Table \ref{tab:def}.

%==================================================
\paragraph{Expression of $A$ (Lines 1--2 of Eq.~\eqref{eq:exp})}
%==================================================
\ 

Unlike existing Pauli-measurement-based methods, Pauli bases are not used to represent $A$ as shown in Eq.~\eqref{eq:Adecomp}.
This means that the Pauli-measurement is no longer used to evaluate the expectation value.
Instead, the real and imaginary parts of $\langle \varphi |b\rangle \langle c| \varphi \rangle$ are directly measured in the present method.
Both the real and imaginary parts of $\langle \varphi |b\rangle \langle c| \varphi \rangle$ can be expressed as the projection of the state $|\varphi\rangle$ on the superposition state of $|b\rangle$ and $|c\rangle$, as shown in the second line of Eq.~\eqref{eq:exp}.

%==================================================
\paragraph{Measurement Operators (Lines 2--3 of Eq.~\eqref{eq:exp})}
%==================================================
A basis transformation is required to evaluate $|(1/\sqrt{2})(\langle b | \pm \langle c|)|\varphi\rangle|^2$ and $|(1/\sqrt{2})(\langle b | \pm {\rm i}\langle c|)|\varphi\rangle|^2$ using measurements in the computational basis.
Fortunately, this basis transformation can be carried out using simple measurement operators defined as
%==============================
\begin{align}
\label{eq:defMRe}
\frac{|b\rangle+|c\rangle}{\sqrt{2}} &:= \big(M^{(b,c)}_{\rm Re}\big)^\dagger |b\rangle \\[5mm]
\label{eq:defMIm}
\frac{|b\rangle+{\rm i}|c\rangle}{\sqrt{2}} &:= \big(M^{(b,c)}_{\rm Im}\big)^\dagger |b\rangle .
\end{align}
%==============================
Concrete circuit representations of these measurement operators are given in Table~\ref{tab:def}.
A nice property of these measurement operators is that each of them contains at most $n-1$ CNOT gates, which is suitable for the NISQ devices.
We refer to the present method as the \red{\emph{extended Bell measurements (XBMs)}} because the measurement operators $M^{(b,c)}_{\rm Re}$ and $M^{(b,c)}_{\rm Im}$ are similar to those for the Bell measurements except for the number of CNOT gates.
%==================================================
\begin{figure}[tb]
\red{
    \centering
    \begin{minipage}{0.45\hsize}
        \centering
        \includegraphics[scale=1]{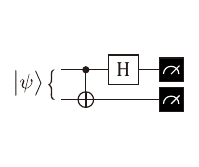}\\
        (a)
    \end{minipage}
    \begin{minipage}{0.45\hsize}
        \centering
        \includegraphics[scale=1]{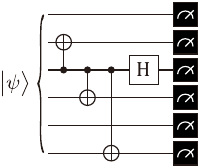}\\
        (b)
    \end{minipage}
    \\[5mm]
    \caption{Typical circuits for (a) standard Bell measurements, and (b) extended Bell measurements.}
    \label{fig:bell}
}
\end{figure}
%==================================================
\red{
In Fig.\ref{fig:bell}, typical circuits for a standard Bell measurements and the extended Bell measurements are shown.
}
%==================================================
\paragraph{Simultaneous Measurements (Lines 3--4 of Eq.~\eqref{eq:exp})}
%==================================================
We proved the following theorem.
\begin{theorem}
\label{theorem:simul}
If $b\oplus c = b' \oplus c'$ and $b'\neq c'$, then $M^{(b,c)}_{\rm Re}=M^{(b',c')}_{\rm Re}$ and $M^{(b,c)}_{\rm Im}=M^{(b',c')}_{\rm Im}$.
\end{theorem}
The proof is given in Appendix~\ref{sec:proof-simul}.
The above theorem states that the number of distinct measurement operators can be drastically reduced.
In addition, the superscript $(\bullet)^{(b,c)}$ can be replaced with $(\bullet)^{(b\oplus c)}$ because each distinct measurement operator is identified by $b\oplus c$.
Reformulation leads to the fourth line of Eq.~\eqref{eq:exp}.

%==================================================
\paragraph{Operator Selecting Model (Lines 4--5 of Eq.~\eqref{eq:exp})}
%==================================================
The fourth line of Eq.~\eqref{eq:exp} does not seem to express an expectation value because it does not take the form $\sum_i a_i P(\hat{a}=a_i)$.
Here, we introduce an operator selecting model $p(l,s)$ that represents the probability of selecting the measurement operator $M^{(l)}_s$, where $l\in\mathcal{S}(A)$ and $s\in \{{\rm Re}, {\rm Im}\}$.
Here, $\mathcal{S}(A)$ is a set of $b\oplus c$ where $A_{bc}$ has non-zero components, as defined in Table~\ref{tab:def}.
Introducing $p(l,s)$, $\langle \varphi | A | \varphi \rangle$ is eventually represented as on the rightmost side of Eq.~\eqref{eq:exp} and takes the form $\sum_i a_i P(\hat{a}=a_i)$.
Although the concrete description of $p(l,s)$ is arbitrary the appropriate definition of $p(l,s)$ reduces the variance of the expectation estimation.

When we regard Eq.~\eqref{eq:exp} as expressing the expectation value in a stochastic process, an event is selecting a measurement operator $M^{(l)}_s$ with probability $p(l,s)$ followed by measuring a bit-string $b$ of a state $M^{(l)}_s|\state\rangle$ in the computational basis, and then obtaining a corresponding outcome.
Denoting a complex random variable corresponding to such an event as $\hat{a}$, we have
%==============================
\begin{align}
{\mathbb E}_{l,s,b\sim p(l,s,b)}[\hat{a}]=\langle \state |A|\state \rangle.
\end{align}
%==============================
In the following, we simply express ${\mathbb E}_{l,s,b\sim p(l,s,b)}[\hat{a}]$ as ${\mathbb E}_{p}[\hat{a}]$.

%==================================================
\paragraph{Significance of Eq.~\eqref{eq:exp}}
%==================================================
There are three important aspects to the expression in Eq.~\eqref{eq:exp}.
\begin{enumerate}
\item There are fewer distinct measurement operators.
\item Each measurement operator has a simple explicit expression.
\item There are fewer CNOT gates in each measurement operator.
\end{enumerate}
Equation~\eqref{eq:exp} contains only \red{$|\mathcal{S}(A)\times \{{\rm Re},{\rm Im}\}|$} measurement operators, $I$, $\{M^{(l)}_{\rm Re}\}_{l \in \mathcal{S}(A)\red{\setminus\{0^n\}}}$, and $\{M^{(l)}_{\rm Im}\}_{l \in \mathcal{S}(A)\red{\setminus\{0^n\}}}$.
The number of distinct measurement operators depends on the structure of $A$, which is discussed further in Sec.~\ref{sec:algorithm}.
A simple explicit expression of each measurement operator is beneficial from the viewpoint of time complexity for circuit construction.
The number of CNOT gates included in each measurement operator is at most $n-1$, which is suitable for the NISQ device because most of the real-device error comes from \red{two-qubit entangling gates}.
In summary, qualitatively speaking, the construction of measurement operators is fast and the evaluation error on a real quantum computer using these operators is expected to be small.
Further details of the complexity analyses are given in Sec.~\ref{sec:algorithm}.

%==================================================
\subsection{Upper Bound of Variance}
%==================================================
An upper bound of the variance of $\hat{a}$ associated with the XBM can be easily calculated as
%==============================
\begin{align}
{\rm Var}_p[\hat{a}]
&={\mathbb E}_p[|\hat{a}|^2] - |{\mathbb E}_p[\hat{a}]|^2\nonumber\\
&\le{\mathbb E}_p[|\hat{a}|^2]\nonumber\\
&=\sum_{l,s,b}|a(l,s,b)|^2p(l,s,b),
\end{align}
%==============================
which depends on the operator selecting model $p(l,s)$.
A simple choice of $p(l,s)$ is a uniform distribution over $l$ and $s$ whereas a more sophisticated choice can reduce the variance of $\hat{a}$ (see Appendix~\ref{app:variance-derivation} for details).

Here, we assume that \red{$|A|^2_{\max}=O({\rm tr}(A^2)/q)$ where $|A|_{\max}:= \max_{b,c}\{|A_{bc}|\}$ is the maximum norm of $A$ and $q$ the number of non-zero components of $A$}, respectively.
Under this assumption, no matter which operator selecting model is used, the upper bound of the variance of $\hat{a}$ can be roughly estimated as
%==============================
\begin{align}
%    {\rm Var}_{p}[\hat{a}] \lessapprox \frac{m^2}{q} {\rm tr}(A^2)
    \red{{\rm Var}_{p}[\hat{a}] \le m^2 |A|^2_{\max} = O \Bigg( \frac{m^2}{q} {\rm tr}(A^2)\Bigg)}
\end{align}
%==============================
where
%==============================
\begin{align}
%    m:=2|\mathcal{S}(A)|+1
    \red{m:=|\mathcal{S}(A) \times \{{\rm Re}, {\rm Im}\}|}
\end{align}
%==============================
is the number of distinct measurement operators (see Appendix~\ref{app:var-shadow} for details).

%==================================================
\subsection{Relationship with the Classical Shadow}
%==================================================
Our method is closely related to the classical shadow with random Clifford measurements~\cite{shadow}.
The main differences are the unitary ensemble and the operator selecting model.
The random Clifford circuits are used in the classical shadow whereas $m$ distinct measurement operators are used in the XBM.
In addition, only the uniform operator selecting model is considered in the original classical shadow whereas an arbitrary operator selecting model is considered in the XBM.

Additionally, the number of \red{two-qubit entangling gates} in each measurement operator is different.
The classical shadow with random Clifford measurements requires $O(n^2/\log(n))$ \red{two-qubit entangling gates} for each measurement operator whereas the XBM requires at most $n-1$ CNOT gates for each.
%Considering the upper bound of the variance presented in the previous subsection, we can say that the XBM is superior to the classical shadow in terms of both the statistical accuracy and the circuit complexity in the case that non-zero components of $A$ are densely filled in a narrow band.
%Such a matrix frequently appear in the case of the FEM as mentioned in the introduction.

\red{
The upper bound of the variance associated with the classical shadow with the random Clifford measurements is $O({\rm tr}(A^2))$~\cite{shadow}.
Hence, when $m^2 \ll q$, our method has an upper bound tighter than that of the classical shadow.
An example of this case is that non-zero components of $A$ are densely filled in the tight bandwidth.
When the bandwidth of $A$ is $k=O(n^c)$, and the non-zero components are densely filled in the band,
%==============================
\begin{align}
    \label{eq:var2}
    {\rm Var}_{p}\big[ \hat{a} \big]
    = O\Bigg(\frac{n^{c+1}}{2^n}{\rm tr}(A^2)\Bigg)
\end{align}
%==============================
(see Appendix~\ref{app:var-shadow} for other cases).
Such a matrix frequently appear in the case of the FEM as mentioned in the introduction.
In this case, the adoption of our method exponentially improves the variance relative to the classical shadow.
}

\red{
The relationship between the XBM and classical shadow is further discussed in Appendix~\ref{app:rel-shadow}, which gives the explicit form of the classical shadow associated with the XBM method.
%This suggest that many properties such as quantum fidelity, two-point correlation functions, entanglement entropy, and others can be predicted using the XBM method through the classical shadow procedure~\cite{shadow}.
}

%==================================================
\section{Algorithm and Complexity}
\label{sec:algorithm}
%==================================================

%==================================================
\begin{algorithm*}[tb]
    \caption{Algorithm for term grouping in the XBM method}
    \label{alg:grouping}
    \SetKwInOut{Input}{input}
    \SetKwInOut{Output}{output}
    \SetKwInOut{Initialize}{initialize}
    \Input{Non-zero matrix elements and corresponding indices, $\{b,c,A_{bc}|A_{bc}\neq 0\}$}
    \Output{Measurement operators $M^{(l)}_s$ and outcomes $\alpha(l,s,b)$}
    \Initialize{$M^{(l)}_s\leftarrow\texttt{None}$ and $\alpha(l,s,b)\leftarrow 0$ for all $l,s,b$}
    \For{$(b,c) \in \{(b,c)|A_{bc}\neq0\}$} {
        $l \leftarrow b \oplus c$\\
        \If{$l\neq0^n$}{
            $j^{(l)}_0\leftarrow \max \{j|l_j=1\}$\\
        }
        \tcc{Compute and store measurement operators} 
        \If{Measurement operators $M^{(l)}_{\rm Re}$ and $M^{(l)}_{\rm Im}$ have not yet been computed}{
            \uIf{$l=0^n$}{
                $M^{(l)}_{\rm Re} \leftarrow I$
            }
            \Else{
                $M^{(l)}_{\rm Re},\,M^{(l)}_{\rm Im} \leftarrow$ (see Table \ref{tab:def}) \\
            }
        }
        \tcc{Compute and store outcomes}
        \uIf{$b=c$}{
            $\alpha(l,{\rm Re},b) \leftarrow \alpha(l,{\rm Re},b) + A_{bc}$ \\
        }
        \uElseIf{$b<c$}{
            $\bar{b} \leftarrow b \oplus 2^{j^{(l)}_0}$\\
            $\alpha(l,{\rm Re},b) \leftarrow \alpha(l,{\rm Re},b) + A_{bc}/2$ \\
            $\alpha(l,{\rm Re},\bar{b}) \leftarrow \alpha(l,{\rm Re},\bar{b}) -  A_{bc}/2$ \\
            $\alpha(l,{\rm Im},b) \leftarrow \alpha(l,{\rm Im},b) +  {\rm i}A_{bc}/2$ \\
            $\alpha(l,{\rm Im},\bar{b}) \leftarrow \alpha(l,{\rm Im},\bar{b}) -  {\rm i}A_{bc}/2$ \\
        }
        \ElseIf{$b>c$}{
            $\bar{c} \leftarrow c \oplus 2^{j^{(l)}_0}$\\
            $\alpha(l,{\rm Re},c) \leftarrow \alpha(l,{\rm Re},c) + A_{bc}/2$ \\
            $\alpha(l,{\rm Re},\bar{c}) \leftarrow \alpha(l,{\rm Re},\bar{c}) -  A_{bc}/2$ \\
            $\alpha(l,{\rm Im},c) \leftarrow \alpha(l,{\rm Im},c) -  {\rm i} A_{bc}/2$ \\
            $\alpha(l,{\rm Im},\bar{c}) \leftarrow \alpha(l,{\rm Im},\bar{c}) +  {\rm i} A_{bc}/2$ \\
        }
    }
\end{algorithm*}
%==================================================

%==================================================
\begin{algorithm*}[tbh]
    \caption{Algorithm for calculating an expectation value in the XBM method. \red{Note that median-of-means or others can be used instead of the arithmetic mean in the last line.}}
    \label{alg:expectation}
    \SetKwInOut{Input}{input}
    \SetKwInOut{Output}{output}
    \SetKwInOut{Initialize}{initialize}
    \Input{Number of shots $N$, operator selecting model $p(l,s)$, measurement operators $\{M^{(l)}_s\}$, quantum state $|\varphi\big>$,  outcomes $\{\alpha(l,s,b)\}$}
    \Output{Expectation value $x$}
%    \Initialize{$x \leftarrow 0$}
%    \Initialize{\red{$x_i \leftarrow 0$ for $i=1,\ldots,N$}}
    \For{$i=1$ \KwTo $N$}{
        \tcc{Pick a measurement operator}
        $l, s \sim p(l,s)$ \tcp*{Classical computing}
        $M \leftarrow M^{(l)}_s$ \tcp*{Access to the classical memory}
        \tcc{Measure a bit string and update $x$}
        $b \sim |\langle b | M | \varphi \rangle|^2$ \tcp*{Quantum computing}
%        $x \leftarrow x + \alpha(l,s,b)$  \tcp*{Access to the classical memory and classical computing}
        $x_i \leftarrow \alpha(l,s,b)$  \tcp*{Access to the classical memory and classical computing}
    }
    \red{$x \leftarrow \sum_i x_i / N$} \tcp*{\red{Take arithmetic mean}}
\end{algorithm*}
%==================================================

The algorithms for term grouping and calculating an expectation value using the XBM method are presented as Algorithms~\ref{alg:grouping} and \ref{alg:expectation}, respectively.
These algorithms follow the equations described in Eq.~\eqref{eq:exp} (see Appendix~\ref{app:derivation} for details).
The greatest advantage of the proposed method, the XBM method, is that the number of unique circuits required to estimate $\big<\psi_0\big|A\big|\psi_1\big>$ can be rapidly reduced.
In this section, the reduction ratio of the number of unique circuits and the computational cost of the XBM method are discussed.

%==================================================
\subsection{Number of Distinct Operators}
%==================================================
In this paper, an $s$-sparse matrix and band matrix with bandwidth $k$ are defined as follows.

\begin{definition}[$s$-sparse matrix $A$~\cite{dewolf}]
$A\in{\mathbb C}^{2^n \times 2^n}$ is $s$-sparse if each column has at most $s$ nonzero entries.
\end{definition}

\begin{definition}[band matrix $A$]
$A\in{\mathbb C}^{2^n \times 2^n}$ is a band matrix with bandwidth $k$ if $A_{ij}=0$ for all $i,j$ such that $|i-j|>k$.
\end{definition}

According to these definitions, the following theorems hold.

\begin{theorem}[$s$-sparse matrix $A$]
    \label{theo:1}
    Let $A\in{\mathbb C}^{2^n\times 2^n}$ and $\big|\psi_0\big>,\big|\psi_1\big>\in{\mathbb C}^{2^n}$. When $A$ is an $s$-sparse matrix, there exists an algorithm such that the number of quantum circuits required to evaluate $\big<\psi_0\big|A\big|\psi_1\big>$ is $2^{n+1}$ in the worst case.
\end{theorem}

\begin{theorem}[band matrix $A$]
    \label{theo:2}
    Let $A\in{\mathbb C}^{2^n\times 2^n}$ and $\big|\psi_0\big>,\big|\psi_1\big>\in{\mathbb C}^{2^n}$. When the bandwidth of $A$ is $k>0$, there exists an algorithm such that the number of quantum circuits required to evaluate $\big<\psi_0\big|A\big|\psi_1\big>$ is $2((n-r)k+2^r)$ in the worst case where $r=\lceil \log_2k \rceil$.
\end{theorem}

Here, we define the upper bound of the number of the unique circuits for an $n$-qubit system with bandwidth $k$ as
%==============================
\begin{align}
\label{eq:upper-m}
\overline{m}(n,k)
&:=\left\{
        \begin{array}{ll}
            2((n-r)k+2^r)  & k>0 \\[5mm]
            1 & k=0
        \end{array}
    \right.
\end{align}
%==============================
\red{where $r:=\lceil \log_2k\rceil$.}
Note that when evaluating $\langle \varphi | A | \varphi \rangle$ whose $A$ has non-zero diagonal components, the upper bound can be reduced to $2((n-r)k+2^r)-1$ for $k>1$ (see Appendix~\ref{app:proofs1}).

The proofs are given in Appendix~\ref{app:proofs}.
Note that a fully dense matrix can be viewed as a $2^n$-sparse matrix or a band matrix with bandwidth $k=2^n-1$.

Unfortunately, we find that the evaluation of the number of circuits in Algorithm~\ref{alg:grouping} is tight; i.e.,  there exists an $s$-sparse matrix that requires $2^{n+1}$ unique circuits, which is the same as the case for the fully dense matrix.
However, if $A$ is a band matrix, then at least the number of unique circuits required to estimate $\big<\psi_0\big|A\big|\psi_1\big>$ can be reduced to $2((n-r)k+2^{r})$, where $k$ is the bandwidth and $r=\lceil \log_2k \rceil$.
When $A$ is sufficiently tight (i.e., bandwidth $k = O(n^c)$), $2((n-r)k+2^{r}) = O(n^{c+1})$.
Note that the number of unique Pauli strings required to evaluate $\big<\psi_0\big|A\big|\psi_1\big>$ using na\"ive Pauli measurements is $2^{n+1}((n-r)k+2^{r})$ (see Sec.~\ref{app:proofs3}).
The XBM thus outperforms the na\"ive Pauli measurements in terms of the number of unique circuits required to evaluate $\big<\psi_0\big|A\big|\psi_1\big>$.
This exponential improvement is due to the difference in the basis of $A$, but the number of unique circuits when using the XBM method does not appreciably increase even if $A$ is represented in the Pauli basis (see Sec.~\ref{app:xbmpauli}).

As mentioned before, our method is the generalization of \cite{poisson1}.
In their formulation, both $\langle \varphi |A^2| \varphi \rangle=\langle \varphi |(B+C^{(0)}+C^{(1)})| \varphi \rangle$ and $\langle b |A| \varphi \rangle$ are calculated to solve the one-dimensional Poisson equation, and the numbers of unique circuits required to evaluate $\langle \varphi |A^2| \varphi \rangle=\langle \varphi |(B+C^{(0)}+C^{(1)})| \varphi \rangle$ and $\langle b |A| \varphi \rangle$ are $4n+1$ and $2n+1$, respectively, where the bandwidth of $B$, $C^{(0)}$, $C^{(1)}$ and $A$ are $2$, $0$, $0$ and $1$, respectively.
This is consistent with $\overline{m}(n,k=2)+\overline{m}(n,k=0)+\overline{m}(n,k=0)=4n+1+1=4n+2$ and $\overline{m}(n,k=1)=2n+2$ (see Eq.~\eqref{eq:upper-m}), where the $+1$ difference comes from the counting of an unnecessary term, ${\rm Im}\big(\langle \varphi | b \rangle \langle b | \varphi \rangle \big)=0$.

%==================================================
\subsection{Time Complexity for Term Grouping}
%==================================================

We assume that data $\{(b,c,A_{bc})|A_{bc}\ne0\}$ are stored in the memory of the classical computer, where $A_{bc}$ is the $(b,c)$ component of the matrix $A$.
For each $(b,c)\in\{(b,c)|A_{bc}\neq0\}:=\mathcal{B}$, bit strings to be measured are calculated from $b$ and $j^{(b\oplus c)}_0=\max\{j|(b\oplus c)_j=1\}$ for $b\neq c$.
The computational cost of calculating $j^{(b\oplus c)}_0$ with classical computers is $O(n)$, and the total time complexity for term grouping is thus $O(nd)$, where $d$ is the size of $\mathcal{B}$, which is the number of non-zero elements of $A$.

Note that in Algorithm~\ref{alg:grouping}, measurement operators $M^{(l)}_s$ are simultaneously calculated with term grouping.
Each measurement circuit can be built in $O(n)$ time (see definition of $M^{(l)}_s$ in Table \ref{tab:def}), and its calculation is only carried out when $M^{(l)}_s$ is not yet calculated. Therefore, the total time complexity of building measurement circuits is $O(nm)$ where $m\le d$ is the number of distinct $b\oplus c$ values among $(b,c)$ pairs with non-zero $A_{bc}$ values, which does not affect the total time complexity of term grouping.

%==================================================
\subsection{Gate Counts for Measurement Circuits}
%==================================================

As described in Table \ref{tab:def}, the maximum number of gate counts is $n+1$ for each unique circuit; i.e., there are $n-1$ CNOT gates, one Hadamard gate and one phase gate.
Hence, the gate count for each measurement circuit is $O(n)$.

%==================================================
\section{Experiments}
\label{sec:experiment}
%==================================================
We implement our model with Qiskit version 0.25.0~\cite{qiskit}.
All experiments are carried out on a single Intel(R) Core(TM) i9-9900X CPU @ 3.50GHz
and \texttt{ibm\_kawasaki} of an IBM-Q system.
To compare with the existing methods, QWC (qubit--wise commuting)~\cite{qwc0,efficient,qwc} and GC (general commuting)~\cite{gc} are conducted with the same machine.
The Python function, \texttt{group\_into\_tensor\_product\_basis\_sets}, implemented in OpenFermion~\cite{openfermion}, is used for QWC whereas 
Python code provided by one of the authors of \cite{gc} is used for GC.
In GC, the Bron-Kerbosch algorithm~\cite{bronkerbosch} is used to solve the minimum clique covering problem.

In our method, the XBM method, ${\rm Re}(\big<\varphi\big|b\big>\big<c\big|\varphi\big>)$  and ${\rm Im}(\big<\varphi\big|b\big>\big<c\big|\varphi\big>)$ can be separately evaluated.
To take this advantage, we consider the two following cases in which ${\rm Im}(\big<\varphi\big|b\big>\big<c\big|\varphi\big>)$ does not need to be estimated.
%==================================================
\begin{enumerate}
    \item $A\in{\mathbb R}^{2^n\times 2^n}$ and only ${\rm Re}(\big<\psi_0\big|A\big|\psi_1\big>)$ is required. For example, $A$ is real symmetric matrix and $\big|\psi_0\big>=\big|\psi_1\big>$.
    \item $\exists \phi_0,\phi_1$ s.t. $e^{{\rm i}\phi_0}\big|\psi_0\big>,e^{{\rm i}\phi_1}\big|\psi_1\big>\in{\mathbb R}^{2^n}$
\end{enumerate}
%==================================================
In these cases, only ${\rm Re}(\big<\varphi\big|b\big>\big<c\big|\varphi\big>)$ needs to be estimated because
%==============================
\begin{align}
    &{\rm Re}\Big(\big<\state\big|A\big|\state\big>\Big) = \sum_{b,c} {\rm Re}\Big(A_{bc}\Big){\rm Re}\Big(\big<\state\big|b\big>\big<c\big|\state\big>\Big) \nonumber \\ 
    &\hspace{13mm}
    - \sum_{b,c} {\rm Im}\Big(A_{bc}\Big){\rm Im}\Big(\big<\state\big|b\big>\big<c\big|\state\big>\Big) \nonumber
\end{align}
%==============================
and
%==============================
\begin{align}
    &{\rm Im}\Big(\big<\state\big|A\big|\state\big>\Big) = \sum_{b,c} {\rm Im}\Big(A_{bc}\Big){\rm Re}\Big(\big<\state\big|b\big>\big<c\big|\state\big>\Big) \nonumber \\ 
    &\hspace{13mm}
    + \sum_{b,c} {\rm Re}\Big(A_{bc}\Big){\rm Im}\Big(\big<\state\big|b\big>\big<c\big|\state\big>\Big). \nonumber
\end{align}
%==============================
The computational cost is reduced by half in these cases and we thus refer these to these cases by ``XBM (half)'' in the following.
Note that in the case that $A\in{\mathbb C}^{2^n\times2^n}$ and $\big|\psi_0\big>,
\big|\psi_1\big>\in{\mathbb C}^{2^n}$, even if $\big|\psi_0\big>=\big|\psi_1\big>$ and $A$ is an Hermitian matrix, which is the case that $\big<\psi_0\big|A\big|\psi_1\big>\in{\mathbb R}$, both ${\rm Re}(\big<\varphi\big|b\big>\big<c\big|\varphi\big>)$ and ${\rm Im}(\big<\varphi\big|b\big>\big<c\big|\varphi\big>)$ have to be evaluated.

$2\big|0\big>\big<1\big|\otimes A$
(Eq.~\eqref{eq:defAprime}) is used instead of $A$ when evaluating the case that $\big|\psi_0\big>\ne\big|\psi_1\big>$.
In QWC and GC, the matrix $A$ is converted to the sum of the Pauli bases before adopting the term grouping algorithms.

%==================================================
\subsection{Term Grouping}
%==================================================

%==================================================
\subsubsection{Setups}
%==================================================
We first show that the present method is more efficient in terms of both the number of unique circuits and the computational cost for term grouping than existing methods.
We conduct the following three experiments:
\begin{enumerate}
\item We compare the number of unique circuits and time required for term grouping between different methods when $\big|\psi_0\big>=\big|\psi_1\big>$. 
\item We repeat the first experiment but for $\big|\psi_0\big>\neq\big|\psi_1\big>$. 
\item We estimate the number of unique circuits for the XBM method when the non-zero elements of $A$ are randomly placed.
\end{enumerate}

In the first and second experiments, all components of $A\in{\mathbb C}^{2^n\times2^n}$ are randomly set to non-zero values.
$A$ is then restricted to be a band matrix by setting $A_{bc}=0$ for $|b-c|>k$, where $k=0,1,\cdots,2^{n-1}$.
In the third experiment, we observe the term groupings while changing the number of non-zero elements, $d$, that are randomly placed in $A$.
The number of unique circuits is obtained 10 times for each value of $d$.
The mean and standard deviation are then calculated.
%==================================================
\subsubsection{Results}
%==================================================
%==================================================
\begin{figure*}[tb]
    \centering
    \subfigure[]{
        \centering
        \includegraphics[width=0.48\linewidth]{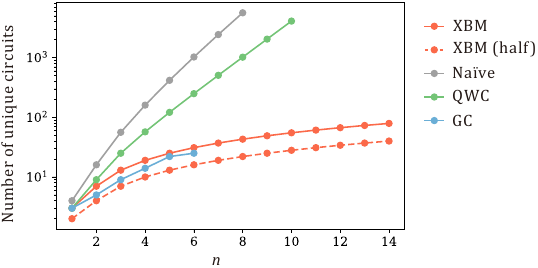}
    \label{fig:nbcircs-band}
    }
    \subfigure[]{
        \centering
        \includegraphics[width=0.48\linewidth]{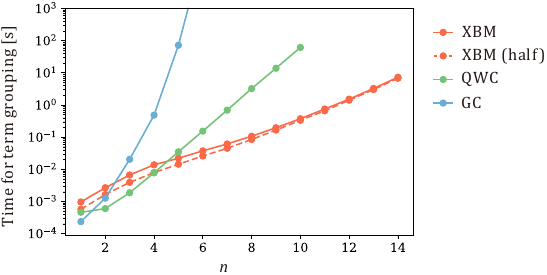}
    \label{fig:time-band}
    }
    \subfigure[]{
        \centering
        \includegraphics[width=0.48\linewidth]{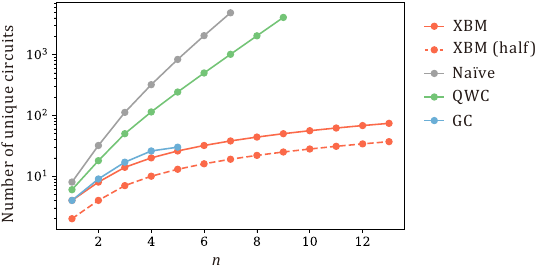}
    \label{fig:nbcircs-band-psi0psi1}
    }
    \subfigure[]{
        \centering
        \includegraphics[width=0.48\linewidth]{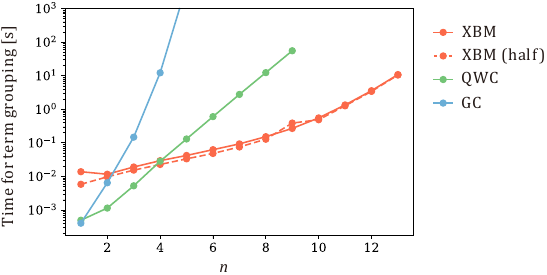}
    \label{fig:time-band-psi0psi1}
    }
    \caption{Number of unique circuits in estimating $\big<\psi_0\big|A\big|\psi_1\big>$ and the time required for term grouping when the bandwidth of $A$ is $k=3$: (a)(b) $\big|\psi_0\big>=\big|\psi_1\big>$; (c)(d) $\big|\psi_0\big>\neq\big|\psi_1\big>$. XBM: present method, XBM (half): present method but only evaluating the real part of $\langle \varphi | b \rangle \langle c | \varphi \rangle$, Na\"ive: Pauli measurements without term grouping, QWC: Pauli measurements with grouping associated with the qubit-wise commutativity~\cite{qwc0,efficient,qwc}, GC: Pauli measurements with grouping associated with the general commutativity~\cite{gc}.}
\end{figure*}
%==================================================
%==================================================
\begin{figure*}[tb]
    \centering
    \subfigure[]{
        \centering
        \includegraphics[width=0.48\linewidth]{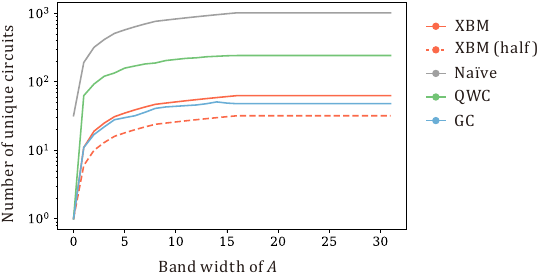}
    \label{fig:nbcircs-5q}
    }
    \subfigure[]{
        \centering
        \includegraphics[width=0.48\linewidth]{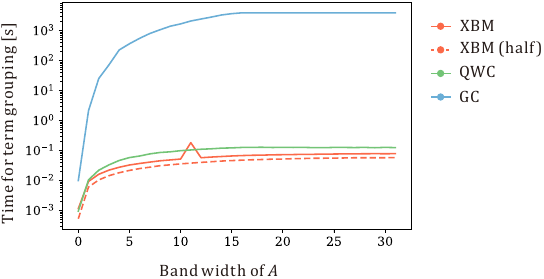}
    \label{fig:time-5q}
    }
    \caption{Number of unique circuits in estimating $\big<\psi_0\big|A\big|\psi_0\big>$ and the time required for term grouping with respect to the bandwidth of $A$, for a five-qubit system. XBM: present method, XBM (half): present method but only evaluating the real part of $\langle \varphi | b \rangle \langle c | \varphi \rangle$, Na\"ive: Pauli measurements without term grouping, QWC: Pauli measurements with grouping associated with the qubit-wise commutativity~\cite{qwc0,efficient,qwc}, GC: Pauli measurements with grouping associated with the general commutativity~\cite{gc}.}
    \label{fig:nbcircs-time-5q}
\end{figure*}
%==================================================
\paragraph{Experiment 1}
Figure~\ref{fig:nbcircs-band} shows the dependency of the number of unique circuits on $n$, in the case that $\big|\psi_0\big>=\big|\psi_1\big>$ and the bandwidth $k=3$.
The XBM and GC results have similar tendencies that are much better than the tendencies of Na\"ive and QWC in terms of the reduction ratio of the unique circuits.
When the imaginary part of $\big<\varphi|b\big>\big<c\big|\varphi\big>$ is not required, as in the case of XBM (half), the number of unique circuits is lower for the XBM method than for the GC method.

Figure~\ref{fig:time-band} shows the dependency of the computational time for term grouping on $n$ when the bandwidth $k=3$.
Obviously, the XBM method is much faster than all other methods for large $n$.
This is because GC and QWC methods heuristically solve a minimum clique cover problem that is NP-hard, whereas the XBM method analytically groups the terms as shown in Sec.~\ref{sec:method}.
Note that this qualitative tendency of the comparison results does not change for any bandwidth as shown in Fig.~\ref{fig:nbcircs-5q} and Fig.~\ref{fig:time-5q}.

\paragraph{Experiment 2}
Figures~\ref{fig:nbcircs-band-psi0psi1} and \ref{fig:time-band-psi0psi1} show the results for the case that $\big|\psi_0\big>\neq\big|\psi_1\big>$.
In this case, regardless of whether the imaginary part of $\big<\varphi\big|b\big>\big<c\big|\varphi\big>$ is required, the number of unique circuits is lower for the XBM method than for all other methods.
This is because the number of unique circuits for $\big|0\big>\big<1\big|\otimes 2A$ is the same as that for $A$ for the XBM method, but not for the other methods.
The tendency of the time required for term grouping against $n$ in the case that $\big|\psi_0\big>\ne\big|\psi_1\big>$ is the same as that in the case that $\big|\psi_0\big>=\big|\psi_1\big>$.

%==================================================
\begin{figure}[tb]
    \centering
    \includegraphics[width=\linewidth]{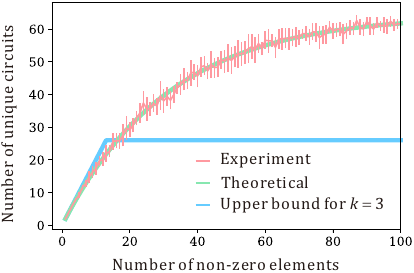}\\
    \caption{Results when the non-zero elements of $A$ are randomly placed for the five-qubit system. Red and green lines are the experimental and analytic results of the XBM, respectively. The blue line is the upper bound when all non-zero elements are placed in a $k=3$ band, $\min(d,\overline{m}(5,3))=\min(d,2((5-\lceil \log_2 3 \rceil)\cdot 3+2^{\lceil \log_2 3 \rceil}))$}
    \label{fig:band-random}
\end{figure}
%==================================================
\paragraph{Experiment 3}
The third experiment is carried out to show that the XBM method is more efficient in terms of the number of unique circuits when $A$ is a band-matrix.
Figure~\ref{fig:band-random} shows the results of the number of unique circuits for the XBM method when the non-zero elements of $A$ are randomly placed; this case is referred to as {\it random-A}.
The upper bound of the number of unique circuits when $A$ is a band-matrix with a bandwidth $k=3$ is also shown in the figure (blue line); this case is referred to as {\it band-A}.

The number of unique circuits required to evaluate {\it random-A} is analytically derived as follows.
Let $A$ have $d$ non-zero elements at random.
Here, we consider a box that contains a total of $N^2$ balls of $N$ types, where $N=2^n$.
In fact, the number of unique circuits is twice the expectation number of ball types when one samples $d$ balls from the box.
The probability that the balls are of type $k$ when $d$ balls are sampled from a box containing a total of $N^2$ balls of $N$ types, $P(k;d,N)$, is
%==================================================
\[\displaystyle{
    P(k;d,N) = \frac{\mathfrak{N}(k,d,N)}{{}_{N^2} {\rm C}_d}
}\]
%==================================================
where
%==================================================
\[\displaystyle{
    \mathfrak{N}(k,d,N) = {}_N{\rm C}_k \ {}_{kN}{\rm C}_d
    \hspace{20mm}
}\]
\[\displaystyle{
    \hspace{20mm}
    - \sum^{k-1}_{i=1}\frac{{}_N{\rm C}_k\  {}_k{\rm C}_i}{{}_N{\rm C}_i} \mathfrak{N}(i,d,N).
}\]
%==================================================
The expectation value is then given as
%==================================================
\[\displaystyle{
    E(d,N) = \sum^{d}_{k=1} k P(k;d,N).
}\]
%==================================================
The green line in Fig.~\ref{fig:band-random} denotes the expectation value $E(d,N=2^5)$ and is consistent with the red line, which is the mean of the experimental results.
The figure shows that the number of unique circuits for {\it random-A} is much larger than that for {\it band-A} when the number of nonzero elements of $A$ is sufficiently large.
For example, when $d=50$, the expectation number of unique circuits for {\it random-A} is approximately 52. 
Meanwhile, the number of unique circuits is at most $26$ when all 50 non-zero elements are put into a bandwidth $k=3$.
This indicates that adoption of the XBM method can reduce the number of unique circuits more efficiently when the bandwidth of $A$ is narrower.

%==================================================
\begin{figure*}[t!]
    \centering
    \subfigure[]{
        \centering
        \includegraphics[scale=1]{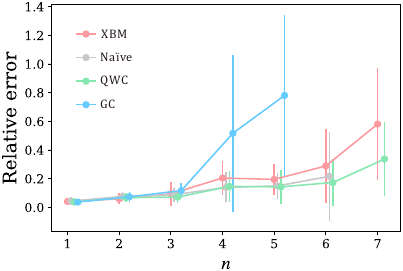}
    \label{fig:real-error-psi0psi0}
    }
    \subfigure[]{
        \centering
        \includegraphics[scale=1]{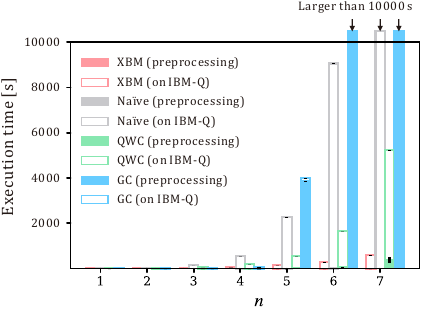}
    \label{fig:real-time-psi0psi0}
    }
    \subfigure[]{
        \centering
        \includegraphics[scale=1]{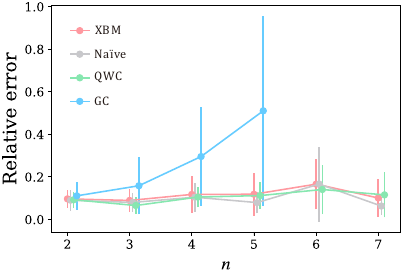}
    \label{fig:real-device-band-error}
    }
    \subfigure[]{
        \centering
        \includegraphics[scale=1]{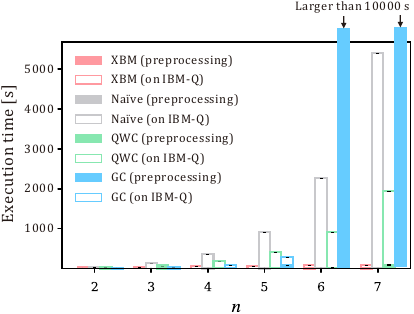}
    \label{fig:real-device-band-time}
    }
    \subfigure[]{
        \centering
        \includegraphics[scale=1]{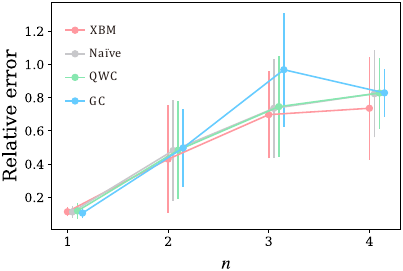}
    \label{fig:real-device-psi0psi1-error}
    }
    \subfigure[]{
        \centering
        \includegraphics[scale=1]{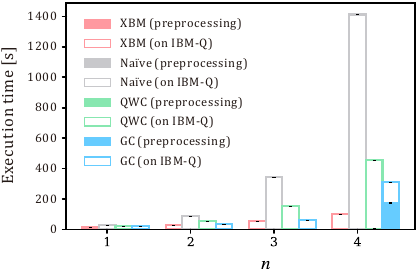}
    \label{fig:real-device-psi0psi1-time}
    }
    \caption{Results of the real device evaluation of $\big<\psi_0\big|A\big|\psi_1\big>$.
    (a)(b) $\big|\psi_0\big>=\big|\psi_1\big>$; (c)(d) $\big|\psi_0\big>=\big|\psi_1\big>$ and $k=3$; (e)(f) $\big|\psi_0\big>\neq\big|\psi_1\big>$. XBM: present method, Na\"ive: Pauli measurements without term grouping, QWC: Pauli measurements with grouping associated with the qubit-wise commutativity~\cite{qwc0,efficient,qwc}, GC: Pauli measurements with grouping associated with the general commutativity~\cite{gc}. {\it method} (preprocessing) and {\it method} (on IBM-Q) denote the execution time of classical preprocessing required before quantum processing and of quantum processing on the IBM-Q system, respectively.}
    \label{fig:real-psi0psi0}
\end{figure*}
%==================================================
%==================================================
\subsection{Real-device Experiments}
%==================================================
%==================================================
\subsubsection{Setups}
%==================================================

In real-device experiments, we use the uniform operator selecting model $p(l,s)=1/m$.
For simplicity, all measurement operators are used at once instead of sampling $l$ and $s$. We set the number of shots for each measurement operator at 8192, which is the maximum number for the IBM-Q system.
We use the \texttt{ibm\_kawasaki} backend for all real-device experiments because its CNOT error and readout error are smaller than those of other backends.
Both the elapsed wall time for preprocessing and execution time on the IBM-Q system are recorded.
Here, preprocessing includes term grouping, calculating the coefficients of Pauli bases, building measurement circuits, and all other calculations conducted on a classical computer before carrying out quantum computing.

We conduct four real device experiments.
\begin{enumerate}
\item We compare the relative errors of $\big<\psi_0\big|A\big|\psi_1\big>$ and execution time for each method when $\big|\psi_0\big>=\big|\psi_1\big>$.
\item We evaluate the relative errors of $\big<\varphi\big|A_{bc}\big|\varphi\big>$ where $A_{bc}$ has only one non-zero element in the $b$th row and $c$th column.
\item We compare the relative errors of $\big<\psi_0\big|A\big|\psi_1\big>$ and execution time for each method when $A$ is a band matrix and $\big|\psi_0\big>=\big|\psi_1\big>$.
\item We repeat the first experiment, but $\big|\psi_0\big>\neq\big|\psi_1\big>$.
\end{enumerate}
Here, we define the relative errors as
\[\displaystyle{
\frac{\Big|\big<\psi_0\big|A\big|\psi_1\big>_{\rm real\ device}-\big<\psi_0\big|A\big|\psi_1\big>_{\rm GT}\Big|}{\Big|\big<\psi_0\big|A\big|\psi_1\big>_{\rm GT}\Big|}
}\]
where $(\bullet)_{\rm real\ device}$ and $(\bullet)_{\rm GT}$ are the evaluation value of $(\bullet)$ on the real device and the ground truth of $(\bullet)$, respectively.

In the first, third and fourth experiments, all the components of $A$ are randomly set as ${\rm Re}(A_{bc}),{\rm Im}(A_{bc})\in[-100,100]$.
In preparing $\big|\psi_0\big>$ and $\big|\psi_1\big>$, $R_x$, $R_y$ and $R_z$ gates are randomly set for each qubit.
The experiments are carried out 10 times for each condition, and the mean and standard deviation of the results are calculated.

In the second experiment, only one component, $A_{bc}$, is set to be non-zero (i.e., ${\rm Re}(A_{bc}),{\rm Im}(A_{bc})\in[-100,100]$), whereas all other components are set to zero.
Unlike the case in the other experiments, $\big|\varphi\big>$ is prepared by acting on all qubits with Hadamard gates to mitigate the effect of the randomness of the $\big|\varphi\big>$ preparation.
This experiment is repeated for $b=0,1,\cdots,2^{n}-1$ and $c=0,1,\cdots,2^{n}-1$.
As a result, the relative errors for all components of $A$ are obtained.
Only XBM and Na\"ive experiments are carried out here.
The experiments are carried out 10 times for each $(b,c)$, and the mean results are calculated.
%==================================================
\subsubsection{Results}
%==================================================
\paragraph{Experiment 1}
As shown in Fig.~\ref{fig:real-error-psi0psi0}, the relative errors of the XBM method are much better than those of the GC method but worse than those of the Na\"ive and QWC methods for $n\ge 4$ when $\big|\psi_0\big>=\big|\psi_1\big>$.
This trend can be explained by the circuit complexity.
The XBM method uses multiple CNOT gates whereas the Na\"ive and QWC methods do not, and the circuit for the GC method is much more complex than the circuits of the other methods.
Although the XBM method is slightly inferior in terms of the relative error compared with the Na\"ive and QWC methods, the total execution time for the XBM method is much shorter than the times for the other methods (Fig.~\ref{fig:real-time-psi0psi0}).
The execution time on a real device for the XBM method is slightly longer than that for the GC method, which depends on the number of unique circuits.
However, the total execution time for the XBM method is much shorter than that for the GC method because the XBM method performs term grouping much more rapidly than the GC method.
In the case of the XBM method, the execution time for term grouping is negligible compared with the execution time on the IBM-Q system.
In summary, the proposed method, namely the XBM method, can evaluate $\big<\psi_0\big|A\big|\psi_0\big>$ with slightly worse accuracy than the Na\"ive and QWC methods but in much less time.

%==================================================
\begin{figure}[tb]
    \centering
    \subfigure[]{
        \centering
        \includegraphics[width=0.8\linewidth]{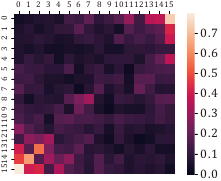}
    \label{fig:elementwise-error-xbm}
    }
    \\
    \subfigure[]{
        \centering
        \includegraphics[width=0.8\linewidth]{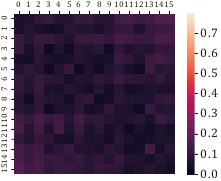}
    \label{fig:elementwise-error-naive}
    }
    \caption{Element-wise relative errors of $\big<\varphi\big|A_{rc}\big|\varphi\big>$ for a four-qubit system evaluated on a real device: (a) XBM method; (b) Na\"ive method.}
    \label{fig:elementwise-error}
\end{figure}
%==================================================
\paragraph{Experiment 2}
Figure~\ref{fig:elementwise-error} compares the elementwise relative error between the XBM and Na\"ive methods. In the case of the XBM method, the relative errors are larger for the anti-diagonal components than for the other components, whereas all components have similar relative errors in the case of the Na\"ive method.
This result is consistent with that of the XBM method using the most CNOT gates in the evaluation of anti-diagonal components.
Note that the number of CNOT gates required for the XBM method to evaluate the $(b,c)$ component of $A$ is
\[\max\Bigg(0, \sum^{2^n-1}_{i=0} b_i \oplus c_i-1\Bigg).\]
This indicates that the XBM method can evaluate $\big<\psi_0\big|A\big|\psi_0\big>$ with much smaller error when $A$ is a band matrix.

\paragraph{Experiment 3}
Figure~\ref{fig:real-device-band-error} shows that the relative errors of the XBM, Na\"ive and QWC methods are small in the case that $A$ is a band matrix; the errors approximately range $0.1-0.15$ and they are all comparable for at least $n \le 7$.
Meanwhile, the XBM method can evaluate $\big<\psi_0\big|A\big|\psi_0\big>$ much more rapidly than the Na\"ive and QWC methods (Fig.~\ref{fig:real-device-band-time}).
It is thus concluded that the XBM method can evaluate $\big<\psi_0\big|A\big|\psi_0\big>$ with the same accuracy as the Na\"ive and QWC methods and in a much shorter time when $A$ is a band matrix.

\paragraph{Experiment 4}
Figures~\ref{fig:real-device-psi0psi1-error} and \ref{fig:real-device-psi0psi1-time} show the results for the case that $\big|\psi_0\big>\neq\big|\psi_1\big>$.
As shown in Fig.~\ref{fig:real-device-psi0psi1-error}, although the XBM method can evaluate $\big<\psi_0\big|A\big|\psi_1\big>$ much more rapidly than all the other methods (Fig.~\ref{fig:real-device-psi0psi1-time}), the relative errors are much larger than in the case that $\big|\psi_0\big>=\big|\psi_1\big>$ for all methods.
This is because in the case that $\big|\psi_0\big>\neq\big|\psi_1\big>$, $\big|\psi_0,\psi_1\big>=\big(\big|0\big>\big|\psi_0\big>+\big|1\big>\big|\psi_1\big>\big)/\sqrt{2}$ is used to evaluate $\big<\psi_0\big|A\big|\psi_1\big>$, which requires a large number of \red{two-qubit entangling} gates as shown in Fig.~\ref{fig:psi0psi1}.
As pointed out above, the evaluation error increases with the number of \red{two-qubit entangling} gates, such as CNOT gates.
Hence, to suppress the error in the case that $\big|\psi_0\big>\neq\big|\psi_1\big>$, reducing the \red{two-qubit entangling} gate error or adopting an alternative approach is required.

%==================================================
\section{Conclusions}
\label{sec:conclusions}
%==================================================
We proposed a method of evaluating $\big<\psi_0\big|A\big|\psi_1\big>$ where $A\in{\mathbb C}^{2^n\times2^n}$ and $\big|\psi_0\big>,\big|\psi_1\big>\in{\mathbb C}^{2^n}$.
Using our method of XBMs, the number of unique circuits required to estimate the expectation value can be reduced to $2(2^r+k(n-r))$ in the worst case, where $k$ is the bandwidth of $A$ and $r:=\lceil \log_2 k \rceil$.
The greatest advantage of the XBM method over existing methods is the computational time required to prepare the simultaneous measurements.
The time complexity of the XBM method including term grouping and building the measurement circuits is $O(nd)$, where $d$ is the number of non-zero elements of $A$ and the gate counts for the measurement operators are $O(n)$ for each unique circuit.
In addition, when non-zero elements of $A$ are densely filled in the bandwidth $k$, the variance of the expectation evaluation is $O((n^{c+1}/2^n){\rm tr} (A^2))$, which is smaller than that of existing methods.
The real device experiments show that although there is a trade-off between accuracy and speed, the XBM method has an appreciable advantage over existing methods in terms of the computational cost, especially when $A$ is a band matrix.

%==================================================
\section{Acknowledgement}
\label{sec:acknowledgement}
%==================================================

This work is partly supported by UTokyo Quantum Initiative.
We thank Prof. Imoto for insightful comments.

%\newpage

%\bibliographystyle{unsrt85}
\bibliographystyle{unsrtnat} 
\bibliography{reference}

\onecolumn\newpage
\appendix

\setcounter{equation}{0}
\renewcommand{\theequation}{S.\arabic{equation}}
\setcounter{figure}{0}
\renewcommand{\thefigure}{S.\arabic{figure}}

%==================================================
\section{Derivation of Eq.~\eqref{eq:exp}}
\label{app:derivation}
%==================================================
%==================================================
\subsection{Deriving Measurement Operators}
%==================================================
Let $A'\in {\mathbb C}^{2^{n} \times 2^{n}}$ be an arbitrary complex matrix and $\big|\psi_0\big>,\big|\psi_1\big>\in{\mathbb{C}^{2^{n}}}$ be arbitrary $n$ qubits quantum states.
The expectation value $\big<\psi_0\big|A'\big|\psi_1\big>$ can be written as
%==============================
\begin{align}
\label{eq:org}
    \big<\psi_0\big|A'\big|\psi_1\big>
    &=\big<\varphi\big|A\big|\varphi\big>,
\end{align}
%==============================
where
%==============================
\begin{align}
\label{eq:defAprime}
    A
    &:=\left\{
        \begin{array}{ll}
            A' & {\rm if}\ \ \ \big|\psi_0\big>=\big|\psi_1\big>\\[5mm]
            \begin{bmatrix} {\bf 0} & 2A' \\ {\bf 0} & {\bf 0} \end{bmatrix} & {\rm else}  
        \end{array}
    \right.
\end{align}
%==============================
and
%==============================
\begin{align}
    \big|\varphi\big>
    &:=\left\{
        \begin{array}{ll}
            \big|\psi_0\big>=\big|\psi_1\big> & {\rm if}\ \ \ \big|\psi_0\big>=\big|\psi_1\big>\\[5mm]
            \displaystyle{\frac{\big|0\big>\big|\psi_0\big>+\big|1\big>\big|\psi_1\big>}{\sqrt{2}}} &{\rm else}  
        \end{array}
    \right. .
\end{align}
%==============================
Note that the state $\Big(\big|0\big>\big|\psi_0\big>+\big|1\big>\big|\psi_1\big>\Big)/\sqrt{2}$ can be prepared using the circuit described in Fig.~\ref{fig:psi0psi1}~\cite{poisson1}.
$A$ is now expressed as
%==============================
\begin{align}
\label{eq:Aexp}
    A=\sum^{2^n-1}_{b=0} \sum^{2^n-1}_{\tilde{c}=0} A_{bc} \big|b\big>\big<c\big|
\end{align}
%==============================
where
%==============================
\begin{align}
\label{eq:tildec}
    c
    &:=\left\{
        \begin{array}{ll}
            \tilde{c} & {\rm if}\ \ \ \big|\psi_0\big>=\big|\psi_1\big>\\[5mm]
            \tilde{c}+2^n &{\rm else}  
        \end{array}
    \right. .
\end{align}
%==============================
Using this representation, $\big<\psi_0\big|A'\big|\psi_1\big>$ is reformulated as
%==============================
\begin{align}
    \big<\psi_0\big|A'\big|\psi_1\big>
    =\sum^{2^n-1}_{b=0} \sum^{2^n-1}_{\tilde{c}=0} A_{bc} \big<\varphi\big| b\big>\big<c\big|\varphi\big>.
\end{align}
%==============================

%==================================================
\begin{figure}[tb]
    \centering
    \includegraphics[scale=1]{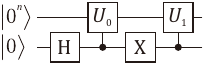}
    \caption{Circuit for preparing $\frac{1}{\sqrt{2}}\Big(\big|0\big>\big|\psi_0\big>+\big|1\big>\big|\psi_1\big>\Big)$~\cite{poisson1}, where $U_i:\big|0^n\big>\mapsto\big|\psi_i\big>$ for $i=0,1$. Note that the most significant bit is illustrated at the bottom of the circuit diagram.}
    \label{fig:psi0psi1}
\end{figure}
%==================================================

Unlike existing Pauli-measurement-based methods, Pauli bases are not used to represent $A$ in the method proposed in this study.
This raises the problem that obtaining the value of $\big<\varphi\big|b\big>\big<c\big|\varphi\big>$ by measurement is no longer straightforward, but there is the advantage that corresponding coefficients $A_{bc}$ can be immediately determined.
Meanwhile, in the case of Pauli-measurement-based methods, the calculation of ${\rm Tr}\big(A^\dagger{\bf P}_i\big)/2^n$ for each Pauli base ${\bf P}_i$ is required.

The purpose of the present method is to evaluate
$\big<\varphi\big|b\big>\big<c \big|\varphi\big>$ for $b,c=0,1,\cdots,2^n-1$ by measurement.
Let $\big|d_\pm^{(b,c)}\big>$ and $\big|e_\pm^{(b,c)}\big>$ be defined as
%==============================
\begin{align}
    \label{def-de}
    \big|d_\pm^{(b,c)}\big>&:=\frac{1}{\sqrt{2}} \Big(\big|b\big>\pm\big|c\big>\Big)\\
    \big|e_\pm^{(b,c)}\big>&:=\frac{1}{\sqrt{2}} \Big(\big|b\big>\pm{\rm i}\big|c\big>\Big).
\end{align}
%==============================
Using these definitions, the real and imaginary parts of $\big<\varphi\big|b\big>\big<c\big|\varphi\big>$ can be expressed as
%==============================
\begin{align}
\label{eq:part-re}
    {\rm Re}\Big(\big<\varphi\big|b\big>\big<c\big|\varphi\big>\Big) &= \frac{1}{2} \Big(\big|D_+^{(b,c)}\big|^2-\big|D_-^{(b,c)}\big|^2\Big)\\
\label{eq:part-im}
    {\rm Im}\Big(\big<\varphi\big|b\big>\big<c\big|\varphi\big>\Big) &= \frac{1}{2} \Big(\big|E_+^{(b,c)}\big|^2-\big|E_-^{(b,c)}\big|^2\Big)
\end{align}
%==============================
where
%==============================
\begin{align}
\label{eq:def-D}
    D_\pm^{(b,c)} &:= \big<d_\pm^{(b,c)} \big|\varphi\big> \in {\mathbb C}\\
\label{eq:def-E}
    E_\pm^{(b,c)} &:= \big<e_\pm^{(b,c)} \big|\varphi\big> \in {\mathbb C}.
\end{align}
%==============================
Equations~\eqref{eq:part-re}-\eqref{eq:def-E} suggest that the real and imaginary parts of $\big<\varphi\big|b\big>\big<c\big|\varphi\big>$ can be evaluated using the probabilities of $\big|\varphi\big>$ evaluated on the basis of $\big|d_\pm^{(b,c)}\big>$ and  $\big|e_\pm^{(b,c)}\big>$, respectively.
Such basis transformations can be achieved by applying the inverse of the preparation unitaries of $\big|d_\pm^{(b,c)}\big>$ and $\big|e_\pm^{(b,c)}\big>$ to $\big|\varphi\big>$, respectively.
Hence, in the following, we focus on how to obtain $\big|d_\pm^{(b,c)}\big>$ and $\big|e_\pm^{(b,c)}\big>$ from $\big|b\big>$ and $\big|c\big>$.

Let $\big|b\big>$ and $\big|c\big>$ be denoted by bit strings as
%==============================
\begin{align}
    \big|b\big>&=\big|b_{n-1}\cdots b_1 b_0\big>, \\
    \big|c\big>&=\big|c_{n-1}\cdots c_1 c_0\big>,
\end{align}
%==============================
where $b=\sum^{n-1}_{i=0}b_i 2^i$, $c=\sum^{n-1}_{i=0}c_i 2^i$, and $b_i, c_i \in \{0,1\}$.
We define two sets of indices as
%==============================
\begin{align}
\label{eq:def-T}
    \mathcal{T}^{(b,c)}&:=\big\{j | b_j \ne c_j,\ j\in{0,1,\cdots,n-1}\big\},\\
\label{eq:def-T0}
    \mathcal{T}_0^{(b)}&:=\big\{j | b_j =0,\ j\in{0,1,\cdots,n-1}\big\}.
\end{align}
%==============================
We then denote an important index by
%==============================
\begin{align}
\label{eq:def-j}
    j^{(b,c)}_0 \in \mathcal{T}^{(b,c)} \cap \mathcal{T}^{(b)}_0.
\end{align}
%==============================
Note that any element of $\mathcal{T}^{(b,c)} \cap \mathcal{T}^{(b)}_0$ can be used for $j^{(b,c)}_0$ in the following.
If $\mathcal{T}^{(b,c)}\cap\mathcal{T}_0^{(b)}=\emptyset$, $\mathcal{T}^{(c,b)}\cap\mathcal{T}_0^{(c)}$ is used instead through the relationships
%==============================
\begin{align}
\label{eq:exchange-re}
    {\rm Re}\Big( \big<\varphi\big| b\big> \big< c \big|\varphi\big>\Big)
    &=  {\rm Re}\Big( \big<\varphi\big| c\big> \big< b \big|\varphi\big>\Big)\\
\label{eq:exchange-im}
     {\rm Im}\Big( \big<\varphi\big| b\big> \big< c \big|\varphi\big>\Big)
    &= - {\rm Im}\Big( \big<\varphi\big| c\big> \big< b \big|\varphi\big>\Big).
\end{align}
%==============================
That is, if there exists $j$ such that $b_j\neq c_j$ but $b_j\neq0$, $b$ and $c$ are exchanged using Eqs.~\eqref{eq:exchange-re} and \eqref{eq:exchange-im}.
Applying the Hadamard gate to $\big|b\big>$ on the $j^{(b,c)}_0$th qubit yields
%==============================
\begin{align}
    {\rm H}(j_0^{(b,c)})\big|b\big>
    &=\frac{1}{\sqrt{2}}
    \big|b_{n-1}\cdots b_{j_0^{(b,c)}+1}\big>\Big(\big|0\big>+\big|1\big>\Big)\big|b_{j^{(b,c)}_0-1}\cdots b_1 b_0\big>\nonumber \\
    &=\frac{1}{\sqrt{2}}\Big(
    \big|b\big>+\big|b_{n-1}\cdots b_{j^{(b,c)}_0+1}c_{j^{(b,c)}_0}b_{j^{(b,c)}_0-1}\cdots b_1 b_0\big>
    \Big),
\end{align}
%==============================
where $c_{j_0^{(b,c)}}=1$ and ${\rm H}(j)$ denotes the Hadamard gate acting on the $j$th qubit.
To convert $\big|b_{n-1}\cdots b_{j_0^{(b,c)}+1}c_{j^{(b,c)}_0}b_{j^{(b,c)}_0-1}\cdots b_1 b_0\big>$ to $\big|c\big>$, multiple CNOT gates are used:
%==============================
\begin{align}
    \prod_{k\in \widetilde{\mathcal{T}}^{(b,c)}}{\rm CNOT}(j^{(b,c)}_0,k) {\rm H}(j_0^{(b,c)})\big|b\big>
    = \frac{1}{\sqrt{2}}\Big(\big|b\big> + \big|c\big>\Big)
    &=\big|d^{(b,c)}_+\big>
\end{align}
%==============================
where 
%==============================
\begin{align}
    \widetilde{\mathcal{T}}^{(b,c)} := \mathcal{T}^{(b,c)} \setminus \{j_0^{(b,c)}\}
\end{align}
%==============================
and ${\rm CNOT}(j,k)$ denotes the CNOT gate whose control and target bits are the $j$th and $k$th qubits, respectively.
When $\widetilde{\mathcal{T}}^{(b,c)}=\emptyset$, the CNOT gate is not used.
We now express $\big|d^{(b,c)}_\pm\big>$ as
%==============================
\begin{align}
    \big|d_\pm^{(b,c)}\big>&=M_{\rm Re}^{(b,c)\dagger} \big|\beta_\pm^{(b,c)}\big>
\end{align}
%==============================
where
%==============================
\begin{align}
\label{eq:def-b}
    \big|\beta_+^{(b,c)}\big>&:=\big|b\big> \\
\label{eq:def-Xb}
    \big|\beta_-^{(b,c)}\big>&:={\rm X}(j_0^{(b,c)})\big|b\big> = \big|b \oplus 2^{j^{(b,c)}_0}\big>.
\end{align}
%==============================
${\rm X}(j)$ indicates the X gate acting on the $j$th qubit, $\oplus$ is a bit-wise exclusive OR operation, and
%==============================
\begin{align}
\label{eq:def-MRe}
    M^{(b,c)\dagger}_{\rm Re} :=
    \prod_{k\in \widetilde{\mathcal{T}}^{(b,c)}}{\rm CNOT}(j_0^{(b,c)},k) {\rm H}(j_0^{(b,c)}).
\end{align}
%==============================
Using the measurement operator $M_{\rm Re}^{(b,c)}$, $\big|D_\pm^{(b,c)}\big|^2$ can be written as
%==============================
\begin{align}
\label{eq:Deq}
    \big|D_\pm^{(b,c)}\big|^2&=\Big|\big<\beta^{(b,c)}_\pm\big |M_{\rm Re}^{(b,c)}\big|\varphi\big>\Big|^2.
\end{align}
%==============================
This value is equivalent to the probabilities of outcomes $\beta^{(b,c)}_\pm$ when measuring $M^{(b,c)}_{\rm Re}\big|\varphi\big>$.

In the same sense, we define $M^{(b,c)}_{\rm Im}$ by
%==============================
\begin{align}
\label{eq:def-MIm}
    M^{(b,c)\dagger}_{\rm Im} :=
    {\rm S}(j^{(b,c)}_0)\prod_{k\in \widetilde{\mathcal{T}}^{(b,c)}}{\rm CNOT}(j_0^{(b,c)},k) {\rm H}(j^{(b,c)}_0)
\end{align}
%==============================
where ${\rm S}(j)$ denotes the phase gate acting on the $j$th qubit.
Then,
%==============================
\begin{align}
    \big|e_\pm^{(b,c)}\big>&=M^{(b,c)\dagger}_{\rm Im} \big|\beta^{(b,c)}_\pm\big>
\end{align}
%==============================
and
%==============================
\begin{align}
    \big|E_\pm^{(b,c)}\big|^2&=\Big|\big<\beta^{(b,c)}_\pm\big|M^{(b,c)}_{\rm Im}\big|\varphi\big>\Big|^2.
\end{align}
%==============================

%The present measurement operators, $M^{(b,c)}_{\rm Re}$ and $M^{(b,c)}_{\rm Im}$, are similar to the operator of the Bell measurement.
%The difference is the number of CNOT gates in the measurement operators; e.g., the basic Bell measurement uses only one CNOT gate (Fig.~\ref{fig:bell}(a)) whereas $M^{(b,c)}_{\rm Re}$ and $M^{(b,c)}_{\rm Im}$ use multiple CNOT gates (Fig.~\ref{fig:bell}(b)).
%Because of this similarity, we refer to the present method as the XBM.

%==================================================
\subsection{Reducing the Number of Circuits}
\label{sec:proof-simul}
%==================================================
To obtain the expectation value, $\big<\psi_0\big|A'\big|\psi_1\big>$, we need to measure the numerous circuits $M^{(0,0)}_{\rm Re}\big|\varphi\big>,M^{(0,1)}_{\rm Re}\big|\varphi\big>,\ldots,M^{(0,0)}_{\rm Im}\big|\varphi\big>,M^{(0,1)}_{\rm Im}\big|\varphi\big>,\ldots$, which requires large computational resources.
Fortunately, some of these circuits can be swapped with each other.
In brief, the following lemma holds.
\begin{lemma}
\label{lemma:simul}
If $b\oplus c=b'\oplus c'$, then any pair in $\big\{|d^{(b,c)}_+\rangle\langle d^{(b,c)}_+|, |d^{(b,c)}_-\rangle\langle d^{(b,c)}_-|, |d^{(b',c')}_+\rangle\langle d^{(b',c')}_+|, |d^{(b',c')}_-\rangle\langle d^{(b',c')}_-|\big\}$ is simultaneously diagonalizable.
In the same way, if $b\oplus c=b'\oplus c'$, then any pair in $\big\{|e^{(b,c)}_+\rangle\langle e^{(b,c)}_+|, |e^{(b,c)}_-\rangle\langle e^{(b,c)}_-|, |e^{(b',c')}_+\rangle\langle e^{(b',c')}_+|, |e^{(b',c')}_-\rangle\langle e^{(b',c')}_-|\big\}$ is simultaneously diagonalizable.
\end{lemma}
\begin{proof}
It is easily confirmed that $|d^{(b,c)}_\pm\rangle\langle d^{(b,c)}_\pm|$ is Hermitian for all $b$ and $c$.
We show that any pair in $\big\{|d^{(b,c)}_+\rangle\langle d^{(b,c)}_+|, |d^{(b,c)}_-\rangle\langle d^{(b,c)}_-|, |d^{(b',c')}_+\rangle\langle d^{(b',c')}_+|, |d^{(b',c')}_-\rangle\langle d^{(b',c')}_-|\big\}$ commutes.
When $b\oplus c=b'\oplus c'$, the following holds.
%==============================
\begin{align*}
b=c \ \ \ &\Rightarrow  \ \ \ b'=b\oplus c\oplus c'=0\oplus c'=c' \\
b'=c'\ \ \ &\Rightarrow  \ \ \ b=b'\oplus c'\oplus c=0\oplus c=c \\
b=c' \ \ \ &\Rightarrow  \ \ \ b'=b\oplus c\oplus c'=0\oplus c=c \\
b'=c \ \ \ &\Rightarrow  \ \ \ b=b'\oplus c'\oplus c=0\oplus c'=c' 
\end{align*}
%==============================
That is,
%==============================
\begin{align*}
b\oplus c=b'\oplus c'
\ \ \ \Rightarrow\ \ \ 
b=c\Leftrightarrow b'={c'}\ \ \ {\rm and}\ \ \ b=c'\Leftrightarrow b'=c .
\end{align*}
%==============================
From this property, we need to check if $\Big[|b\rangle \langle c|,|b'\rangle \langle c'|\Big]=0$ or not, where $[A,B]=AB-BA$ is a commutator, only in the cases of (i) $b\neq b'$ and $c\neq c'$ and (ii) $b=b'$ and $c=c'$.
In the same sense, we need to check if $\Big[|b\rangle \langle c|,|c'\rangle \langle b'|\Big]=0$ or not only in the cases of (i) $b\neq c'$ and $c\neq b'$, (ii) $b=c'$ and $c=b'$.
These checks are easily carried out and the result is that
%==============================
\begin{align*}
b\oplus c=b'\oplus c'
\ \ \ \Rightarrow\ \ \ 
\Big[|b\rangle \langle c|,|b'\rangle \langle c'|\Big]
=\Big[|b\rangle \langle c|,|c'\rangle \langle b'|\Big]
=0.
\end{align*}
%==============================
As a consequence, because $|d^{(b,c)}_\pm\rangle=(|b\rangle\pm|c\rangle)/\sqrt{2}$,
%==============================
\begin{align*}
&b\oplus c=b'\oplus c'
\ \ \ \Rightarrow\ \ \ 
[A,B]=0 \nonumber \\
&\ \ \ {\rm where}\ \ \ 
A,B \in \Big\{|d^{(b,c)}_+\rangle\langle d^{(b,c)}_+|, |d^{(b,c)}_-\rangle\langle d^{(b,c)}_-|, |d^{(b',c')}_+\rangle\langle d^{(b',c')}_+|, |d^{(b',c')}_-\rangle\langle d^{(b',c')}_-|\Big\}
\end{align*}
%==============================
Any pair in $\big\{|d^{(b,c)}_+\rangle\langle d^{(b,c)}_+|, |d^{(b,c)}_-\rangle\langle d^{(b,c)}_-|, |d^{(b',c')}_+\rangle\langle d^{(b',c')}_+|, |d^{(b',c')}_-\rangle\langle d^{(b',c')}_-|\big\}$ is Hermitian and commutes, and is therefore simultaneously diagonalizable.
In the same way, it can be shown that $\big\{|e^{(b,c)}_+\rangle\langle e^{(b,c)}_+|, |e^{(b,c)}_-\rangle\langle e^{(b,c)}_-|, |e^{(b',c')}_+\rangle\langle e^{(b',c')}_+|, |e^{(b',c')}_-\rangle\langle e^{(b',c')}_-|\big\}$ are simultaneously diagonalizable.
\qed
\end{proof}

Lemma~\ref{lemma:simul} states that there exists an unitary that simultaneously diagonalizes $|d^{(b,c)}_\pm\rangle\langle d^{(b,c)}_\pm|$ and $|d^{(b',c')}_\pm\rangle\langle d^{(b',c')}_\pm|$ when $b \oplus c = b' \oplus c'$.
This could appreciably reduce the number of unique circuits required to estimate the expectation value.
However, lemma~\ref{lemma:simul} does not state that which unitary can do.
In the following, we prove that $M^{(b,c)}_{\rm Re}$ can simultaneously diagonalize $|d^{(b,c)}_\pm\rangle\langle d^{(b,c)}_\pm|$ and $|d^{(b',c')}_\pm\rangle\langle d^{(b',c')}_\pm|$ when $b \oplus c = b' \oplus c'$.
In the same sense, we can show that $M^{(b,c)}_{\rm Im}$ can simultaneously diagonalize $|e^{(b,c)}_\pm\rangle\langle e^{(b,c)}_\pm|$ and $|e^{(b',c')}_\pm\rangle\langle e^{(b',c')}_\pm|$ when $b \oplus c = b' \oplus c'$.

Before proving the general case, let us consider an example.
Consider the example of Eq.~\eqref{eq:part-re}; i.e., 
\[\displaystyle{
{\rm Re}\Big(\big<\varphi\big|011\big>\big<101\big|\varphi\big>\Big)
=\frac{1}{2}\Big(\big|D^{(011,101)}_+\big|^2-\big|D^{(011,101)}_-\big|^2\Big).
}\]
In this case, because $b=011$ and $c=101$,
\[\displaystyle{
\mathcal{T}^{(b,c)}\cap\mathcal{T}^{(b)}_0=\{1,2\}\cap\{2\}=\{2\}
}\]
yielding
\[
j^{(b,c)}_0=2
\]
as in Eq.~\eqref{eq:def-j}.
The bit strings used are 
\[
\beta^{(b,c)}_+=b=011
\ \ \ {\rm and}\ \ \ 
\beta^{(b,c)}_-={\rm X}(2)b=111.
\]
Hence, from Eq.~\eqref{eq:Deq},
\[\displaystyle{
\big|D^{(011,101)}_+\big|^2=\Big|\big<011\big|M^{(011,101)}_{\rm Re}\big|\varphi\big>\Big|^2
\ \ \ {\rm and}\ \ \ 
\big|D^{(011,101)}_-\big|^2=\Big|\big<111\big|M^{(011,101)}_{\rm Re}\big|\varphi\big>\Big|^2,
}\]
where
\[\displaystyle{
M^{(011,101)}_{\rm Re}={\rm H}(2){\rm CNOT}(2,1).
}\]
We easily find that these values are equivalent to $\big|D^{(101,011)}_\pm\big|^2=\Big|\big<\beta^{(101,011)}_\pm\big| M^{(101,011)}_{\rm Re}\big|\varphi\big>\Big|^2$, in which $M^{(101,011)}_{\rm Re}={\rm H}(1){\rm CNOT}(1,2)$.
This means that $M^{(011,101)}_{\rm Re}$ is no longer required to evaluate ${\rm Re}\Big(\big<\varphi\big|011\big>\big<101\big|\varphi\big>\Big)$ because $M^{(101,011)}_{\rm Re}$ can be used instead.

Here, the above example is generalized.
\settheoremtag{\ref{theorem:simul}}
\begin{theorem}
{\rm (reprint)}
If $b\oplus c = b' \oplus c'$ and $b'\neq c'$, then $M^{(b,c)}_{\rm Re}=M^{(b',c')}_{\rm Re}$ and $M^{(b,c)}_{\rm Im}=M^{(b',c')}_{\rm Im}$.
\end{theorem}
\begin{proof}
First of all, Eq.~\eqref{eq:exchange-re} and \eqref{eq:exchange-im} guarantee that $\big<\varphi\big|b\big>\big<c\big|\varphi\big>$ can be evaluated with $\big<\varphi\big|c\big>\big<b\big|\varphi\big>$ instead.
We can therefore restrict the following discussion to the case that $b<c$.
The index $j^{(b,c)}_0$ is then redefined as
%==============================
\begin{align}
\label{eq:def-max-j}
    j^{(b,c)}_0 = \mathop{{\rm max}}_{j} \mathcal{T}^{(b,c)} \cap \mathcal{T}^{(b)}_0,
\end{align}
%==============================
which satisfies Eq.~\eqref{eq:def-j}.
We also consider pairs of all $(b,c)$ and $(b',c')$ that satisfy $b\oplus c=b'\oplus c'$.
These conditions give the equality
%==============================
\begin{align}
\label{eq:property-M-1}
    \forall s\in\{{\rm Re},{\rm Im}\},\ \ \ 
    M^{(b',c')}_s = M^{(b,c)}_s \ \ \ {\rm for\ \ \ }b\oplus c=b'\oplus c' \ {\rm and }\ b'<c'.
\end{align}
%==============================
Equation~\eqref{eq:property-M-1} is proved as follows.
The bit-wise exclusive OR operation, $b\oplus c$, returns bit strings that are composed of values of $1$ when $b_j\neq c_j$ and $0$ otherwise; e.g., $011 \oplus 101 = 110$.
Therefore, the relationship $b\oplus c=b'\oplus c'$ yields $\mathcal{T}^{(b,c)}=\mathcal{T}^{(b',c')}$, given Eq.~\eqref{eq:def-T}, and obviously
%==============================
\begin{align}
    j^* = \mathop{{\rm max}}_{j} \mathcal{T}^{(b,c)} = \mathop{{\rm max}}_{j} \mathcal{T}^{(b',c')}.
\end{align}
%==============================
Furthermore, the conditions $b<c$ and $b'<c'$ provide $(b_{j^*},c_{j^*})=(b'_{j^*},c'_{j^*})=(0,1)$, and thus $j^*\in\mathcal{T}^{(b)}_0$ and $j^*\in\mathcal{T}^{(b')}_0$.
Consequently, we obtain $j^{(b,c)}_0=j^{(b',c')}_0$.
The measurement operators $M^{(b,c)}_s$ depend only on $j^{(b,c)}_0$ and $\mathcal{T}^{(b,c)}$ (see Eqs.~\eqref{eq:def-MRe} and \eqref{eq:def-MIm}), resulting in Eq.~\eqref{eq:property-M-1}.
In the case that $b'>c'$, as discussed above, the symmetry with respect to the exchange of $b$ and $c$ provides
%==============================
\begin{align}
\label{eq:property-M-2}
    \forall s\in\{{\rm Re},{\rm Im}\},\ \ \ 
    M^{(b',c')}_s = M^{(c,b)}_s \ \ \ {\rm for\ \ \ }b\oplus c=b'\oplus c' \ {\rm and }\ b'>c'.
\end{align}
%==============================
\qed
\end{proof}

In short, by combining the properties of Eqs.~\eqref{eq:exchange-re} and \eqref{eq:exchange-im} and Eqs.\eqref{eq:property-M-1} and \eqref{eq:property-M-2}, all $\big<\varphi\big|b'\big>\big<c'\big|\varphi\big>$ that satisfy $b'\oplus c'=b\oplus c$ can be evaluated using the same measurement operator $M^{(b,c)}_s$.

%==================================================
\subsection{Closed-form Expression of $\langle \varphi | A | \varphi \rangle$}
%==================================================
Combining the equations derived in the previous subsections, $\langle \varphi | A | \varphi \rangle$ can be written as
%==============================
\begin{align}
    \label{eq:close1}
    \langle \varphi | A | \varphi \rangle
    &=\sum_b A_{bb} \langle \varphi | b \rangle \langle b | \varphi \rangle \nonumber \\
    &\hspace{10mm}+\sum_b \sum_{c>b} A_{bc} {\rm Re} \big( \langle \varphi | b \rangle \langle c | \varphi \rangle \big)
    +\sum_b \sum_{c<b} A_{bc} {\rm Re} \big( \langle \varphi | b \rangle \langle c | \varphi \rangle \big) \nonumber \\
    &\hspace{10mm}+\sum_b \sum_{c>b} {\rm i}A_{bc} {\rm Im} \big( \langle \varphi | b \rangle \langle c | \varphi \rangle \big)
    +\sum_b \sum_{c<b} {\rm i}A_{bc} {\rm Im} \big( \langle \varphi | b \rangle \langle c | \varphi \rangle \big)\nonumber \\
    &=\sum_b A_{bb} \langle \varphi | b \rangle \langle b | \varphi \rangle \nonumber \\
    &\hspace{10mm}+\sum_b \sum_{c>b} \Bigg[ \big(A_{bc}+A_{cb}\big) {\rm Re} \big( \langle \varphi | b \rangle \langle c | \varphi \rangle \big)+ {\rm i}\big(A_{bc}-A_{cb}\big) {\rm Im} \big( \langle \varphi | b \rangle \langle c | \varphi \rangle \big) \Bigg]\nonumber \\
    &=\sum_b A_{bb} \langle \varphi | b \rangle \langle b | \varphi \rangle
    +\sum_b \sum_{c>b} \Bigg[ 2 A_{(bc)} {\rm Re} \big( \langle \varphi | b \rangle \langle c | \varphi \rangle \big)+ 2{\rm i} A_{[bc]} {\rm Im} \big( \langle \varphi | b \rangle \langle c | \varphi \rangle \big)\Bigg] \nonumber \\
    &= \sum_b A_{bb} | \langle b | \varphi \rangle |^2 \nonumber \\
    &{\hspace{10mm}}+\sum_b \sum_{c>b}A_{(bc)} \Big( |\langle b | M^{(b,c)}_{\rm Re} | \varphi \rangle|^2 - |\langle b \oplus 2^{j^{(b,c)}_0}| M^{(b,c)}_{\rm Re} | \varphi \rangle|^2\Big)\nonumber\\
    &{\hspace{10mm}}+\sum_b \sum_{c>b}{\rm i}A_{[bc]} \Big( |\langle b | M^{(b,c)}_{\rm Im} | \varphi \rangle|^2 - |\langle b \oplus 2^{j^{(b,c)}_0}| M^{(b,c)}_{\rm Im} | \varphi \rangle|^2\Big)
\end{align}
%==============================
Here, we introduce
%==============================
\begin{align}
    l := b \oplus c.
\end{align}
%==============================
Using this notation, $j^{(b,c)}_0$ can be rewritten as
%==============================
\begin{align}
    j^{(b,c)}_0
    &=\max\{j|b_j=0,c_j=1\}\nonumber\\
    &=\max\{j|b_j=0,(l\oplus b)_j=1\}\nonumber\\
    &=\left\{
        \begin{array}{ll}
            \max\{j|l_j=1\} & {\rm for}\ \ \ c>b \\
            \max\{j|b_j=0,(l\oplus b)_j=1\} & {\rm for}\ \ \ c<b
        \end{array}
    \right..
\end{align}
%==============================
Note that $c=l\oplus b$.
Because $c>b$ throughout Eq.~\eqref{eq:close1}, $j^{(b,c)}_0$ can be rewritten as
%==============================
\begin{align}
    j^{(l)}_0 := \max\{j|l_j=1\}.
\end{align}
%==============================
In addition, because $M^{(b',c')}_{s}=M^{(b,c)}_{s}$ when $b'\oplus c'=b\oplus c=l$, where $s\in\{{\rm Re},{\rm Im}\}$, $M^{(b,c)}_s$ can be rewritten as
%==============================
\begin{align}
    M^{(l)}_{s} := M^{(b,c)}_{s}.
\end{align}
%==============================
Moreover, the following lemma holds.
%==================================================
\begin{lemma}
\label{lemma}
Let $j^{(l)}_0:=\max\{j|l_j=1\}$ and $\bar{b}:=b\oplus 2^{j^{(l)}_0}$.
Then, $\bar{b}\oplus l<\bar{b}$ if $b\oplus l > b$.
\end{lemma}
\begin{proof}
If $b_{j^{(l)}_0}=1$, $b\oplus l<b$ because the operation $b\oplus l$ changes $b_{j^{(l)}_0}$ from $1$ to $0$, where $j^{(l)}_0$ is the most significant bit that changes value via the operation $b\oplus l$.
Therefore, $b_{j^{(l)}_0}=0$.
This results in $(b\oplus l)_{j^{(l)}_0}=1$, $(b\oplus l \oplus 2^{j^{(l)}_0})_{j^{(l)}_0}=0$, and $(b\oplus 2^{j^{(l)}_0})_{j^{(l)}_0}=1$.
In addition, $(b\oplus l)_{>j^{(l)}_0}=b_{>j^{(l)}_0}$ because $l_{>j^{(l)}_0}=0$ and $x\oplus 0=x$.
This means that the magnitude relationships between $b\oplus l$ and $b$ is determined by the value of the $j^{(l)}_0$th bit, which is the same for the relationship between $b\oplus l\oplus 2^{j^{(l)}_0}$ and $b\oplus 2^{j^{(l)}_0}$.
Therefore, $b\oplus l\oplus 2^{j^{(l)}_0}<b\oplus 2^{j^{(l)}_0}$.
\qed
\end{proof}
%==================================================
According to Lemma~\ref{lemma},
%==============================
\begin{align}
\{l|b\oplus l>b\} = \{l|\bar{b}\oplus l<\bar{b}\}
\ \ \ {\rm where}\ \ \ 
\bar{b}:=b\oplus 2^{j^{(l)}_0},
\end{align}
%==============================
and obviously,
%==============================
\begin{align}
\{b\}=\{0,1,\cdots,2^n-1\}=\{\bar{b}\}.
\end{align}
%==============================
Using the above notations and relationships, we have
%==============================
\begin{align}
    \label{eq:reRe}
    \sum_b \sum_{c>b}A_{(bc)} |\langle b \oplus 2^{j^{(b,c)}_0}| M^{(b,c)}_{\rm Re} | \varphi \rangle|^2
    &=\sum_b \sum_{l:b\oplus l>b}A_{(b\,b\oplus l)}|\langle b \oplus 2^{j^{(l)}_0}| M^{(l)}_{\rm Re} | \varphi \rangle|^2 \nonumber \\
    &=\sum_{{b}} \sum_{l:{b}\oplus l<{b}} A_{({b}\oplus 2^{j^{(l)}_0}\,{b}\oplus l\oplus 2^{j^{(l)}_0})}|\langle {b}| M^{(l)}_{\rm Re} | \varphi \rangle|^2
\end{align}
%==============================
and
%==============================
\begin{align}
    \label{eq:reIm}
    \sum_b \sum_{c>b}A_{[bc]} |\langle b \oplus 2^{j^{(b,c)}_0}| M^{(b,c)}_{\rm Im} | \varphi \rangle|^2
    &=\sum_b \sum_{l:b\oplus l>b}A_{[b\,b\oplus l]}|\langle b \oplus 2^{j^{(l)}_0}| M^{(l)}_{\rm Im} | \varphi \rangle|^2 \nonumber \\
    &=\sum_{{b}} \sum_{l:{b}\oplus l<{b}} A_{[{b}\oplus 2^{j^{(l)}_0}\,{b}\oplus l\oplus 2^{j^{(l)}_0}]}|\langle {b}| M^{(l)}_{\rm Im} | \varphi \rangle|^2.
\end{align}
%==============================
Note that all $\bar{b}$ and $b$ are exchanged.
Substituting Eqs.~\eqref{eq:reRe} and \eqref{eq:reIm} into \eqref{eq:close1}, we obtain
%==============================
\begin{align}
\label{eq:exp-close}
    \langle \varphi | A | \varphi \rangle
    &=\sum_b A_{bb} \langle \varphi | b \rangle \langle b | \varphi \rangle \nonumber \\
    &\hspace{10mm}+ \sum_b \sum_{l\in\mathcal{S}(A)\red{\setminus\{0^n\}}}\Bigg( (-1)^w A_{(b'c')} |\langle b | M^{(l)}_{\rm Re} | \varphi \rangle|^2
    + {\rm i}(-1)^w A_{[b'c']} |\langle b | M^{(l)}_{\rm Im} | \varphi \rangle|^2\Bigg),
\end{align}
%==============================
where
\red{
%==============================
\begin{align}
    \mathcal{S}(A)
%    &:=\{l|b\oplus l>b\} \cup \{l|b\oplus l<b\}\nonumber \\
%    &=\{b \oplus c | b \neq c\}
    &:=\{l|b\oplus l>b\} \cup \{l|b\oplus l<b\} \cup \{0^n\}\nonumber \\
    &=\{b \oplus c\}
\end{align}
%==============================
}
and
%==============================
\begin{align}
    (b',c',w)
    &=\left\{
        \begin{array}{ll}
            (b,b\oplus l,0) & {\rm for}\ \ \ b<b\oplus l \\
            (b\oplus 2^{j^{(l)}_0}, b\oplus l\oplus 2^{j^{(l)}_0}, 1) & {\rm for}\ \ \ b>b\oplus l
        \end{array}
    \right..
\end{align}
%==============================

%==================================================
\subsection{Operator Selecting Model}
%==================================================
At first glance, the right-hand side of Eq.~\eqref{eq:exp-close} does not seem to express the expectation value because it does not take the form $\sum_i a_i P(\hat{a}=a_i)$.
Here, we further reformulate Eq.~\eqref{eq:exp-close}.
During the reformulation, we introduce an operator selecting model that is useful in evaluating the variance of Eq.~\eqref{eq:exp-close}.

Consider that the measurement operator $M^{(l)}_s$ is selected with probability $p(l,s)$ where \red{$l\in\mathcal{S}(A)$}, $s\in\{{\rm Re}, {\rm Im}\}$, $M^{(0^n)}_s=I$ and $p(l=0^n)=p(l=0^n,s)$.
The values of
%==============================
\begin{align}
p(l,s,b)
    :=\left\{
        \begin{array}{ll}
            p(l,s)|\langle b | \state \rangle |^2 & l=0^n \\
            p(l,s)|\langle b | M^{(l)}_{\rm Re}| \state \rangle |^2 & l\neq 0^n,\ s={\rm Re}  \\
            p(l,s)|\langle b | M^{(l)}_{\rm Im}| \state \rangle |^2 & l\neq 0^n,\ s={\rm Im}
        \end{array}
    \right.
\end{align}
%==============================
are probabilities because $p(l,s,b)\in[0,1]$ and $\sum_{l,s,b} p(l,s,b)=1$.
Note that in this representation,
%==============================
\begin{align}
p(b|l,s)
    :=\left\{
        \begin{array}{ll}
            |\langle b | \state \rangle |^2 & l=0^n \\
            |\langle b | M^{(l)}_{\rm Re}| \state \rangle |^2 & l\neq 0^n,\ s={\rm Re}  \\
            |\langle b | M^{(l)}_{\rm Im}| \state \rangle |^2 & l\neq 0^n,\ s={\rm Im}
        \end{array}
    \right. .
\end{align}
%==============================
Meanwhile, the coefficients
%==============================
\begin{align}
\label{eq:coeff2}
a(l,s,b):=\left\{
        \begin{array}{ll}
            \displaystyle{ \frac{A_{bb}}{p(l,s)} } & l=0^n \\[5mm]
            \displaystyle{ \frac{(-1)^w A_{(b'c')}}{p(l,s)} } & l\neq 0^n,\ s={\rm Re}  \\[5mm]
            \displaystyle{ \frac{{\rm i}(-1)^w A_{[b'c']}}{p(l,s)} } & l\neq 0^n, \ s={\rm Im}
        \end{array}
    \right. .
\end{align}
%==============================
can be seen as outcomes.
With these definitions, Eq.~\eqref{eq:exp} can be rewritten as
%==============================
\begin{align}
\langle \state |A|\state \rangle
= \sum_{l,s,b} a(l,s,b) \, p(l,s,b),
\end{align}
%==============================
which is mathematically in the form of an expectation value.

%==================================================
\section{Derivation of Variance}
\label{app:variance-derivation}
%==================================================
According to Eq.~\eqref{eq:exp}, the upper bound of the variance associated with the XBM can be written as
%==============================
\begin{align}
{\rm Var}[\hat{a}]
&\le \sum^{2^n-1}_{b=0} \Bigg|\frac{A_{bb}}{p(l=0^n)} \Bigg|^2 p(l=0^n,b)\nonumber\\
&+ \sum_{l\in \mathcal{S}(A)\red{\setminus\{0^n\}}}\sum^{2^{n}-1}_{b=0} \Bigg( \Bigg| \frac{(-1)^w A_{(b'c')}}{p(l,s={\rm Re})} \Bigg|^2 p(l,s={\rm Re},b) + \Bigg|\frac{{\rm i}(-1)^w A_{[b'c']}}{p(l,s={\rm Im})}\Bigg|^2 p(l,s={\rm Im},b)\Bigg).
\end{align}
%==============================
When we assume 
%==============================
\begin{align}
p(l,s)=p_{\rm uniform}(l,s)=\frac{1}{m}
\end{align}
%==============================
where
\red{
%==============================
\begin{align}
    m:=|\mathcal{S}(A) \times \{{\rm Re}, {\rm Im}\}|
\end{align}
%==============================
}
is the number of distinct measurement operators appearing in Eq.~\eqref{eq:exp}, we obtain
%==============================
\begin{align}
{\rm Var}_{p_{\rm uniform}}[\hat{a}]
&\le m^2 \Bigg(\sum^{2^n-1}_{b=0} |A_{bb}|^2 p(l=0^n,b)\nonumber\\
&\hspace{10mm}+ \sum_{l\in \mathcal{S}(A)\red{\setminus\{0^n\}}}\sum^{2^{n}-1}_{b=0} \Big(|A_{(b'c')}|^2 p(l,s={\rm Re},b) + |A_{[b'c']}|^2 p(l,s={\rm Im},b)\Big)\Bigg) \nonumber \\
&\le m^2 \max_{b,c}\{|A_{bb}|^2,|A_{(bc)}|^2,|A_{[bc]}|^2\}\nonumber\\
&= m^2 \max_{b,c}\{|A_{(bc)}|^2,|A_{[bc]}|^2\}.
\end{align}
%==============================
Meanwhile, when we assume
%==============================
\begin{align}
p(l,s)=p_{\rm weighted}(l,s):=\left\{
        \begin{array}{ll}
            \displaystyle{ \frac{\max_b |A_{bb}|}{Z} } & l=0^n \\[5mm]
            \displaystyle{ \frac{\max_b |A_{(b'c')}|}{Z} } & l\neq 0^n,\ s={\rm Re}  \\[5mm]
            \displaystyle{ \frac{\max_b |A_{[b'c']}|}{Z} } & l\neq 0^n, \ s={\rm Im}
        \end{array}
    \right. ,
\end{align}
%==============================
where
%==============================
\begin{align}
\label{eq:defZ}
    Z := \max_b |A_{bb}|+ \sum_{l \in \mathcal{S}(A)\red{\setminus\{0^n\}}} \Big( \max_b |A_{(b',c')}| + \max_b |A_{[b',c']}|\Big),
\end{align}
%==============================
we obtain
%==============================
\begin{align}
{\rm Var}_{p_{\rm weighted}}[\hat{a}]
&\le |Z|^2 \Bigg[\sum^{2^n-1}_{b=0} \Bigg|\frac{A_{bb}}{\max_{\tilde{b}} A_{\tilde{b}\tilde{b}}} \Bigg|^2 p(l=0^n,b)\nonumber\\
&\hspace{10mm}+ \sum_{l\in \mathcal{S}(A)\red{\setminus\{0^n\}}}\sum^{2^{n}-1}_{b=0} \Bigg( \Bigg| \frac{A_{(b'c')}}{\max_{\tilde{b}'}A_{(\tilde{b}'\tilde{c}')}} \Bigg|^2 p(l,s={\rm Re},b) + \Bigg|\frac{A_{[b'c']}}{\max_{\tilde{b}'}A_{[\tilde{b}'\tilde{c}']}}\Bigg|^2 p(l,s={\rm Im},b)\Bigg)\Bigg]\nonumber\\
&\le |Z|^2 \max\Bigg\{
\Bigg(\frac{A_{bb}}{\max_{\tilde{b}} A_{\tilde{b}\tilde{b}}} \Bigg)^2,\ 
\Bigg(\frac{A_{(bc)}}{\max_{\tilde{b},\tilde{c}}A_{(\tilde{b}\tilde{c})}} \Bigg)^2,\ 
\Bigg(\frac{A_{[bc]}}{\max_{\tilde{b},\tilde{c}}A_{[\tilde{b}\tilde{c}]}} \Bigg)^2
\Bigg\} \nonumber \\
&\le |Z|^2.
\end{align}
%==============================
Obviously, ${\rm Var}_{p_{\rm weighted}}[\hat{a}] \le {\rm Var}_{p_{\rm uniform}}[\hat{a}]$.

%==================================================
\section{Variance Comparison}
\label{app:variance-comparison}
%==================================================
%==================================================
\subsection{Pauli measurements}
%==================================================
In \cite{xu2021variational}, $A$ is expressed as $A=\sum_{i,j}A_{ij}|i\rangle \langle j|$, where $|i\rangle \langle j|$ is expressed using the Pauli basis as
%==============================
\begin{align}
    |i \rangle \langle j|=\sum_k \frac{s^{(i,j)}_k\sigma^{(i,j)}_k}{2^n},
    \ \ \ 
    s^{(i,j)}_k \in \{\pm1,\pm{\rm i}\},
    \ \ \ 
    \sigma^{(i,j)}_k \in \{I,X,Y,Z\}^{\otimes n},
\end{align}
%==============================
where $I$, $X$, $Y$, and $Z$ are the Pauli matrices.
For example,
\[\displaystyle{
|010\rangle\langle110|
=|0\rangle\langle1|\otimes |1\rangle\langle1|\otimes |0\rangle\langle0|
=\frac{X-{\rm i}Y}{2}\otimes \frac{I-Z}{2}\otimes \frac{I+Z}{2}
}\]
\[\displaystyle{
=\frac{1}{2^3}(XII+XIZ-XZI-XZZ-{\rm i}YII-{\rm i}YIZ+{\rm i}YZI+{\rm i}YZZ).
}\]
The expectation value of the Pauli string can be evaluated using the Pauli measurement as
%==============================
\begin{align}
\langle \varphi | \sigma^{(i,j)}_k | \varphi \rangle = \sum^{2^n-1}_{b=0} t^{(i,j)}_{k,b}\big|\langle b | U^{(i,j)}_k |\varphi \rangle\big|^2,\ \ \ t^{(i,j)}_{k,b}\in\{\pm1\},
\end{align}
%==============================
where $U^{(i,j)}_k$ is an unitary for the Pauli measurement required to estimate $\langle \varphi | \sigma^{(i,j)}_k | \varphi \rangle$, which consists of the Hadamard and phase gates.
Substituting with these equations, $\langle \varphi |A| \varphi \rangle$ can be written as
%==============================
\begin{align}
\langle \varphi | A | \varphi \rangle = \sum_{i,j} \langle \varphi | \Big(A_{ij}|i\rangle\langle j|\Big) | \varphi \rangle
= \sum_{i,j}  \sum_k \sum^{2^n-1}_{b=0} \frac{A_{ij} s^{(i,j)}_k t^{(i,j)}_{k,b} \big|\langle b | U^{(i,j)}_k |\varphi\rangle\big|^2}{2^n}
\end{align}
%==============================

Assuming
%==============================
\begin{align}
\label{eq:model-pauli}
p(i,j)=\frac{|A_{ij}|}{\sum_{p,q}|A_{pq}|},\ \ \ 
p(k|i,j)=\frac{1}{2^n},\ \ \ 
p(b|i,j,k)= \big|\langle b | U^{(i,j)}_k |\psi\rangle\big|^2,
\end{align}
%==============================
$\langle \varphi | A | \varphi \rangle$ can be rewritten as
%==============================
\begin{align}
\langle \varphi | A | \varphi \rangle
= \sum_{i,j}  \sum_k \sum^{2^n-1}_{b=0} \frac{A_{ij} \sum_{p,q}|A_{pq}|s^{(i,j)}_k t^{(i,j)}_{k,b}}{|A_{ij}|}  p(i,j)p(k|i,j)p(b|i,j,k).
\end{align}
%==============================
Hence, the variance is
%==============================
\begin{align}
    {\rm Var}_{p}[\hat{a}]
    &\le \sum_{i,j}  \sum_k \sum^{2^n-1}_{b=0} \left| \frac{A_{ij} \sum_{p,q}|A_{pq}|s^{(i,j)}_k t^{(i,j)}_{k,b}}{|A_{ij}|}\right|^2  p(i,j)p(k|i,j)p(b|i,j,k) \nonumber \\
    &= \Big(\sum_{p,q}|A_{pq}|\Big)^2
\end{align}
%==============================
Note that the method proposed in \cite{xu2021variational} is closely related to the classical shadow with the random Pauli measurements.
That is, a classical shadow, $\hat{\rho}$, associated with \cite{xu2021variational} can be expressed as
%==============================
\begin{align}
    {\rm tr}(A\hat{\rho}) = s^{(i,j)}_kt^{(i,j)}_{k,b}\sum_{p,q}|A_{pq}|\frac{A_{ij}}{|A_{ij}|},
\end{align}
%==============================
which is a slightly modified version of the classical shadow with the random Pauli measurements.
Concretely, \cite{xu2021variational} modifies the operator selecting model $p(i,j)p(k|i,j)$ from a uniform distribution to Eq.~\eqref{eq:model-pauli}.

It is easily confirmed that 
%==============================
\begin{align}
\label{eq:varcomp1}
\max_b |A_{bb}|+ \sum_{l \in \mathcal{S}(A)\red{\setminus\{0^n\}}} \Big( \max_b |A_{(b',c')}| + \max_b |A_{[b',c']}| \Big) \le \sum_{p,q}|A_{pq}|
\end{align}
%==============================
because the left-hand side of Eq.~\eqref{eq:varcomp1}, which is the same as the definition of $Z$ (Eq.~\eqref{eq:defZ}), only sums over some elements of $A$ whereas the right hand side of Eq.~\eqref{eq:varcomp1} sums over all elements of $A$.
Thus, the upper bound of the variance associated with the XBM is tighter than that of \cite{xu2021variational} and the classical shadow with the random Pauli measurements.

%==================================================
\subsection{Classical Shadow}
\label{app:var-shadow}
%==================================================
The classical shadow~\cite{shadow} is expected to evaluate the expectation value from few measurements.
In \cite{shadow}, two methods, namely random Pauli measurements and random Clifford measurements, were proposed.
With respect to the arbitrary matrix $A$, the random Clifford measurements are far more suitable than the random Pauli measurements because the random Clifford measurements do not require the decomposition of $A$ into the Pauli basis.
The upper bound of the variance associated with the random Clifford measurements obeys
%==============================
\begin{align}
    \label{eq:var-shadow}
    {\rm Var}\big[ \hat{a} \big]_{{\rm shadow}}
    \le\sqrt{9+\frac{6}{2^n}}{\rm tr}(A^2)
    =O\big({\rm tr}(A^2)\big).
\end{align}
%==============================

In the XBM method, the upper bound of the variance can be written as
%==============================
\begin{align}
    \label{eq:xbmbound}
    {\rm Var}_{p}\big[ \hat{a} \big]_{\rm XBM}
    &\le m^2 |A|_{\max}^2
\end{align}
%==============================
for both models of $p(l,s)$ proposed in the main body, where $|A|_{\max}=\max_{b,c}\big\{|A_{bc}|\big\}$ is a maximum norm of $A$.
To compare Eq.~\eqref{eq:var-shadow} with Eq.~\eqref{eq:xbmbound}, we assume 
\red{
%==============================
\begin{align}
    \label{eq:assumption}
    |A|^2_{\max}
    =O\Bigg(\frac{{\rm tr}(A^2)}{q}\Bigg)
    \ \ \ ({\rm assumption})
\end{align}
%==============================
where $q$ is the number of non-zero elements of $A$.
This assumption is rephrased as that the square of the maximum value of the elements of $A$ is the same order as the mean of the square of the non-zero elements of $A$, and is valid for such as FEM because the element stiffness of the system is the same order in each element.
Under this assumption,
}
%
\if0
%==============================
\begin{align}
    \label{eq:assumption}
    |A|^2_{\max}
    \approx |\mu_A|^2+|\sigma_A|^2
    \ \ \ ({\rm assumption})
\end{align}
%==============================
where $\mu_A$ and $\sigma_A$ are respectively the mean and standard deviation of non-zero components of $A$
\footnote{
\red{
For example, if the non-zero components of $A$ obey logistic distribution $f(x;\mu,s)=\exp(-\frac{x-\mu}{s})/(s(1+\exp(-\frac{x-\mu}{s}))^2)$, the mode of the maximum value of $A$ is $\mu+s\log(N_{\rm sample})$ where $N_{\rm sample}$ is the number of samples obeying the distribution $f(x;\mu,s)$.
This comes $O(|A|^2_{\max})=O(|\mu+s\log(N_{\rm sample})|^2)\approx O(|\mu|^2+|\sigma|^2)$ where $\sigma=\pi s/\sqrt{3}$ is the standard deviation of $f(x;\mu,s)$, which is roughly comparable to Eq.~\eqref{eq:assumption}.
}
}.
In addition, we consider that $A$ is a band matrix whose bandwidth is $k$.
Moreover, from the definition of the variance,
%==============================
\begin{align}
    \label{eq:relationship}
    {\rm tr}(A^2)=q \big(|\mu_A|^2+|\sigma_A|^2\big)
\end{align}
%==============================
where $q$ is the number of non-zero elements of $A$.
Substituting Eqs.~\eqref{eq:assumption} and \eqref{eq:relationship} to Eq.~\eqref{eq:xbmbound}, we obtain
\fi
%
%==============================
\begin{align}
    {\rm Var}_{p}\big[ \hat{a} \big]_{\rm XBM}
%    &\lessapprox \frac{m^2}{q}{\rm tr}(A^2)
    &\red{=O\Bigg( \frac{m^2}{q}{\rm tr}(A^2) \Bigg)}
\end{align}
%==============================
This means the variance when using the XBM depends on the ratio of $m^2$ to $q$ under the assumption of Eq.~\eqref{eq:assumption}.
The range of the number of non-zero elements, $q$, is
%==============================
\begin{align}
    m \le q \le 2^n(2k+1)-k(k+1).
\end{align}
%==============================
The lower bound, $q=m$, is the case that each group for simultaneous measurement has only one element (Fig.~\ref{fig:qlowhigh}(a)) whereas the upper bound, $q= 2^n(2k+1)-k(k+1)$, is the case that all $A_{ij}$ for $|i-j|\le k$ have a non-zero value (Fig.~\ref{fig:qlowhigh}(b)).
%==================================================
\begin{figure}[tb]
    \centering
    \begin{minipage}{0.45\hsize}
        \centering
        \includegraphics[scale=0.7]{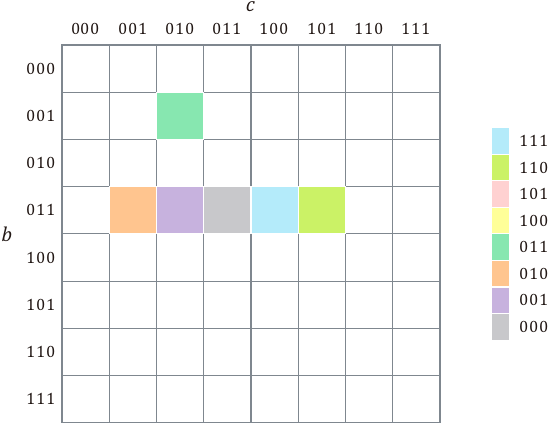}\\
        (a)
    \end{minipage}
    \begin{minipage}{0.45\hsize}
        \centering
        \includegraphics[scale=0.7]{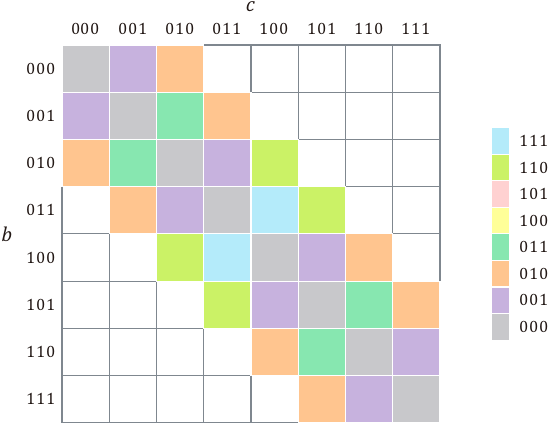}\\
        (b)
    \end{minipage}
    \\[5mm]
    \caption{Cases for the lower and upper bounds of the number of non-zero elements $q$ for $n=3$ qubits and $k=2$ bandwidth. (a) Lower-bound case $q=m$. (b) Upper-bound case $q=2^n(2k+1)-k(k+1)$. Components of the same color can be simultaneously evaluated by one measurement operator using the XBM method. The color range is defined by $l=b\oplus c$, where $b$ and $c$ are the row and column indices of $A$.}
    \label{fig:qlowhigh}
\end{figure}
%==================================================
The worst case for the variance when using the XBM method is $q=m=2((n-\lceil\log_2k\rceil)k+2^{\lceil\log_2k\rceil})\approx2k(1+n-\log_2k)$.
In this case,
%==============================
\begin{align}
    \frac{m^2}{q}
    =\left\{
        \begin{array}{ll}
            \displaystyle{O(n^{c+1})} & {\rm when}\ \ \ k=O(n^{c})\\[5mm]
            O(2^n) & {\rm when}\ \ \ k=O(2^{n})
        \end{array}
    \right. .
\end{align}
%==============================
Meanwhile, when $q=2^n(2k+1)-k(k+1)=O(k2^n)$,
%==============================
\begin{align}
    \frac{m^2}{q}
    =\left\{
        \begin{array}{ll}
            \displaystyle{O\Bigg(\frac{n^{c+1}}{2^n}\Bigg)} & {\rm when}\ \ \ k=O(n^{c})\\[5mm]
            O(1) & {\rm when}\ \ \ k=O(2^n)
        \end{array}
    \right. .
\end{align}
%==============================
In summary,
%==============================
\begin{align}
    \label{eq:xbmbound2}
%    O\Big( {\rm Var}_{p}\big[ \hat{a} \big]_{\rm XBM} \Big)
    \red{ {\rm Var}_{p}\big[ \hat{a} \big]_{\rm XBM} }
    =\left\{
        \begin{array}{ll}
            O\big(n^{c+1}\,{\rm tr}(A^2)\big) & {\rm when}\ \ \ k=O(n^c)\ \ \ {\rm and} \ \ \ q=O(m) \\[5mm]
            O\big(2^n\,{\rm tr}(A^2)\big) & {\rm when}\ \ \ k=O(2^n)\ \ \ {\rm and} \ \ \ q=O(m) \\[5mm]
            \displaystyle{O\Bigg(\frac{n^{c+1}}{2^n}\,{\rm tr}(A^2)\Bigg)} & {\rm when}\ \ \ k=O(n^{c})\ \ \ {\rm and} \ \ \ q=O(k2^n) \\[5mm]
            O\big(\,{\rm tr}(A^2)\big) & {\rm when}\ \ \ k=O(2^n)\ \ \ {\rm and} \ \ \ q=O(k2^n) 
        \end{array}
    \right. .
\end{align}
%==============================

When $q=O(m)$, which is the case shown in Fig.~\ref{fig:qlowhigh}(a), the variance associated with the XBM is worse than that associated with the classical shadow (Eqs.~\eqref{eq:xbmbound2}$_1$ and \eqref{eq:xbmbound2}$_2$).
However, when $q=O(k2^n)$, which is the case shown in Fig.~\ref{fig:qlowhigh}(b), the variance associated with the XBM is at least of the same order as that associated with the classical shadow (Eqs.~\eqref{eq:xbmbound2}$_3$ and \eqref{eq:xbmbound2}$_4$).
In particular, when $k=O(n^c)$ and $q=O(k2^n)$, which is the case that $A$ is a matrix with a narrow bandwidth whose non-zero elements are densely filled in the band, the variance associated with the XBM is exponentially tighter than that associated with the classical shadow (Eqs.~\eqref{eq:xbmbound2}$_3$).
Such $A$ frequently appear as the global stiffness matrix in the finite element method as mentioned in the introduction of the main body.

%==================================================
\section{Classical Shadow with XBM}
\label{app:rel-shadow}
%==================================================
The XBM method is closely related to the classical shadow~\cite{shadow}.
Let $\rho$, $\mathcal{U}$ and $\mathcal{M}$ be an arbitrary density matrix, unitary ensemble, and quantum channel, respectively.
These are defined as
\red{
%==============================
\begin{align}
\label{eq:ensemble-xbm}
%    \mathcal{U}:=\{I\}\cup\{M^{(l)}_s\}_{l\in\mathcal{S}(A),s\in\{{\rm Re},{\rm Im}\}}
    \mathcal{U}:=\{M^{(l)}_s\}_{l\in\mathcal{S}(A),s\in\{{\rm Re},{\rm Im}\}}
\end{align}
%==============================
}
and
%==============================
\begin{align}
    \mathcal{M}(\rho) 
    & := \mathbb{E}_{U\sim{\rm P}_\mathcal{U},\,b\sim\langle b|\rho|b\rangle}\big[U^\dagger |b\rangle \langle b|U\big] \nonumber \\
    &=\sum_b\Bigg[ p(l=0^n) \langle b | \rho | b \rangle |b \rangle \langle b |
    +\sum_{l\in\mathcal{S}(A)\red{\setminus\{0^n\}}} \sum_{s\in\{{\rm Re},{\rm Im}\}} p(l,s) \langle b | M^{(l)}_s \rho M^{(l)\dagger}_s | b \rangle M^{(l)\dagger}_s | b \rangle \langle b | M^{(l)}_s
    \Bigg],
\end{align}
%==============================
respectively, where ${\rm P}_{\mathcal U}$ is a probability distribution over $\mathcal{U}$.
A classical shadow, $\hat{\rho}:= \mathcal{M}^{-1}\big(U^\dagger | b \rangle \langle b | U\big)$, $U\sim {\rm P}_\mathcal{U}$, $b\sim\langle b|\rho|b\rangle$, associated with the XBM method is expressed as
%==============================
\begin{align}
\label{eq:shadow-xbm}
    \hat{\rho}(l,s,b)
    &=\left\{
        \begin{array}{ll}
            \displaystyle{\frac{|b \rangle \langle b|}{p(l=0^n)}} & {\rm for}\ \ \ l=0^n\\[5mm]
            \displaystyle{(-1)^w \frac{|c' \rangle \langle b'|+|b' \rangle \langle c'|}{2p(l,s={\rm Re})}} & {\rm for}\ \ \ l\neq0^n\ \ \ {\rm and}\ \ \ s={\rm Re} \\[5mm]
            \displaystyle{{\rm i}(-1)^w \frac{|c' \rangle \langle b'|-|b' \rangle \langle c'|}{2p(l,s={\rm Im})}} & {\rm for}\ \ \ l\neq0^n\ \ \ {\rm and}\ \ \ s={\rm Im}
        \end{array}
    \right.
\end{align}
%==============================
where $(b',c',w)$ are defined in Table~\ref{tab:def}.
The above can be easily derived by replacing $A_{bc}$ with $|c\rangle \langle b|$ because
$|\psi\rangle\langle\psi|=\sum_{b,c}|c\rangle \langle b| \langle \psi | b \rangle \langle c | \psi \rangle$
takes the same form as Eq.~\eqref{eq:Adecomp}.
The simple choice for the operator selecting model is
%==============================
\begin{align}
    p(l=0^n)=p(l\neq0^n,s)=\frac{1}{2^{n+1}-1}
\end{align}
%==============================
where $2^{n+1}-1$ is the number of unique circuits used to evaluate the expectation value of the dense $2^n\times 2^n$ matrix.
The classical shadow associated with the XBM method satisfies
%==============================
\begin{align}
    \mathbb{E}_{l,s,b\sim p(l,s,b)}[\hat{\rho}(l,s,b)] = \rho
\end{align}
%==============================
and
%==============================
\begin{align}
    {\rm tr}(A\hat{\rho}(l,s,b)) = a(l,s,b)
\end{align}
%==============================
where $a(l,s,b)$ is defined in Eq.~\eqref{eq:coeff2}.
There is a low requirement of the classical memory to store the classical shadow associated with the XBM method.
That is, only the indices $\{(l,s)\}$ and bit strings measured $\{b\}$ need to be stored.

The measurement operators associated with the XBM method are Clifford circuits, which consist of CNOT, Hadamard and phase gates, and therefore the XBM can be seen as the classical shadow with the random Clifford measurement whose unitary ensemble is a subset of Clifford group ${\rm Cl}(2^n)$ (Eq.~\eqref{eq:ensemble-xbm}), and the unitary is selected with probability $p(l,s)$ that can be modeled to reduce the variance.
From this point of view, our conclusion can be restated that by using an unitary ensemble Eq.~\eqref{eq:ensemble-xbm} and appropriate operator selecting model, the expectation value can be evaluated computationally more efficiently than by using random Clifford measurements in a specific condition.
Concretely, 
\begin{enumerate}
\item the number of \red{two-qubit entangling gates} in each measurement operator is reduced from $O(n^2/\log(n))$ to $O(n)$ and
\item the variance reduces from $O({\rm tr}(A^2))$ to $O((n^{c+1}/2^n){\rm tr}(A^2))$ when non-zero elements of $A$ are densely filled in the bandwidth $k$.
\end{enumerate}

%==================================================
\section{Proofs}
\label{app:proofs}
%==================================================

%==================================================
\subsection{Colored Matrix}
\label{app:colored}
%==================================================
%==================================================
\begin{figure}[tb]
    \centering
    \includegraphics[scale=0.7]{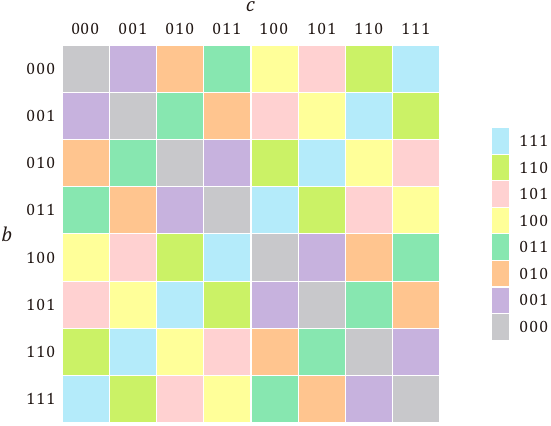}\\
    \caption{Schematic illustration of the colored matrix in the case of the three-qubit system. The $(b,c)$ component is colored with the $b\oplus c$~th color. Components of the same color can be simultaneously measured using the XBM method.}
    \label{fig:colored}
\end{figure}
%==================================================

In proving theorems \ref{theo:1} and \ref{theo:2}, we introduce a colored matrix (see Fig.~\ref{fig:colored}) as a visual reference.
In the colored matrix, the $(b,c)$ component of $A$ is positioned on the $b$~th row and $c$~th column, and each component is colored with the $b\oplus c$~th color.
Components of the same color are simultaneously measurable using the XBM method.
Hence, the number of colors with a non-zero element value is the same as the number of unique circuits used in evaluating $\big<\psi_0\big|A\big|\psi_1\big>$.

%==================================================
\subsection{Proof of Theorem \ref{theo:1}}
\label{app:proofs1}
%==================================================

%==================================================
\begin{figure}[tb]
    \centering
    \begin{minipage}{0.23\hsize}
        \centering
        \includegraphics[scale=1]{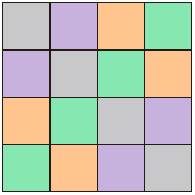}\\
        (a)
    \end{minipage}
    \begin{minipage}{0.23\hsize}
        \centering
        \includegraphics[scale=1]{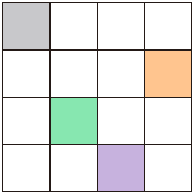}\\
        (b)
    \end{minipage}
    \begin{minipage}{0.23\hsize}
        \centering
        \includegraphics[scale=1]{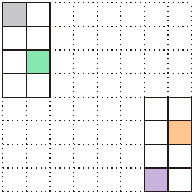}\\
        (c)
    \end{minipage}
    \\[5mm]
    \begin{minipage}{0.23\hsize}
        \centering
        \includegraphics[scale=1]{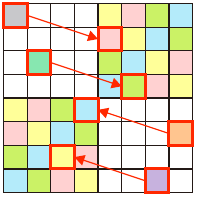}\\
        (d)
    \end{minipage}
    \begin{minipage}{0.23\hsize}
        \centering
        \includegraphics[scale=1]{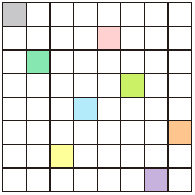}\\
        (e)
    \end{minipage}
    \begin{minipage}{0.23\hsize}
        \centering
        \includegraphics[scale=1]{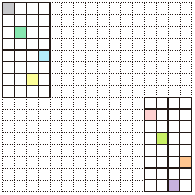}\\
        (f)
    \end{minipage}
    \\[5mm]
    \begin{minipage}{0.23\hsize}
        \centering
        \includegraphics[scale=1]{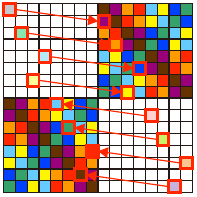}\\
        (g)
    \end{minipage}
    \begin{minipage}{0.23\hsize}
        \centering
        \includegraphics[scale=1]{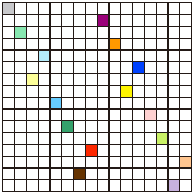}\\
        (h)
    \end{minipage}
    \begin{minipage}{0.23\hsize}
        \centering
        \includegraphics[scale=1]{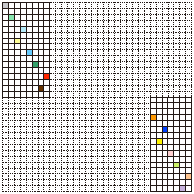}\\
        (i)
    \end{minipage}
    \\[5mm]
    \caption{Step-by-step algorithm for making a 1-sparse matrix including all colors: (a) the same illustration as in Fig.~\ref{fig:colored} for $n=2$; (b) in the first step, manually obtain a 1-sparse matrix for $n=2$ including all colors; (c) split the matrix (b) in half and move the halves to the top-left and bottom-right; (d) fill the top-right and bottom-left color, copy the patterns from the top-left and bottom right to the top-right and bottom-left, respectively, and delete the unmarked square's color; (e) 1-sparse matrix for $n=3$ including all colors; (g)-(i) repeat (c) to (f) until the matrix has the desired size.}
    \label{fig:sparse}
\end{figure}
%==================================================

We show that there exists an algorithm that produces a 1-sparse matrix including all colors in the colored matrix defined in Sec.~\ref{app:colored}.
First, for the two-qubit system, a 1-sparse matrix including all four colors is found manually (Fig.~\ref{fig:sparse}(a) and (b)).
Next, a $2^2\times2^2$ matrix is split in half and the halves are moved to the top-left and right-bottom of a $2^3\times 2^3$ matrix as in Fig.~\ref{fig:sparse}(c).
The colored patterns are then copied to the top-right and bottom-left as in Fig.~\ref{fig:sparse}(d).
Repeating this procedure as in Fig.~\ref{fig:sparse}(e) to (i), a 1-sparse matrix including all colors is obtained for any size $2^n \times 2^n$.

\qed

%==================================================
\subsection{Proof of Theorem \ref{theo:2}}
\label{app:proofs2}
%==================================================

%==================================================
\begin{figure}[tb]
    \centering
    \includegraphics[scale=1]{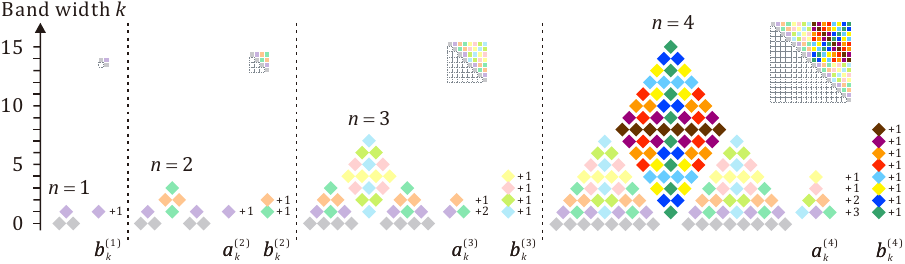}\\
    \caption{Schematic illustration of the number of colors of the colored matrix for $n=1,2,3,4$.}
    \label{fig:band-nbcircs}
\end{figure}
%==================================================

%==================================================
\begin{figure}[tb]
    \centering
    \includegraphics[scale=1]{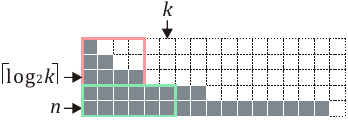}\\
    \caption{Schematic illustration of matrix $b^{(j)}_i$ where the horizontal and vertical correspond to indices $i$ and $j$, respectively. A filled square indicates an entry of $1$ whereas a blank indicates an entry of $0$. The number of entries included in the $n\times k$ square is equal to $\sum^{n}_{j=1}\sum^k_{i=1}b^{(j)}_i$, which can be decomposed into the red and green square corresponding to $\sum^{\lceil\log_2k\rceil}_{j=1}2^{j-1}$ and $(n-\lceil\log_2k\rceil)k$, respectively.}
    \label{fig:bmatrix}
\end{figure}
%==================================================

The number of colors in a colored matrix of bandwidth $k$, which we denote by $N^{(n)}_{k}$, is counted as in Fig.~\ref{fig:band-nbcircs}.
From Fig.~\ref{fig:band-nbcircs}, $N^{(n)}_{k}$ can be written as 
%==============================
\begin{align}
N^{(n)}_{k}=1+\sum^k_{i=1}\Big( a^{(n)}_i+b^{(n)}_i\Big)
\end{align}
%==============================
where
%==============================
\begin{align}
a^{(1)}_i=0,
\hspace{5mm}
b^{(n)}_i=
\left\{
\begin{array}{ll}
1 & 1\le i \le 2^{n-1}\\
0 & {\rm else}
\end{array}
\right.,
\hspace{5mm}
a^{(n)}_i=a^{(n-1)}_i+b^{(n-1)}_i.
\end{align}
%==============================
Solving this recurrence relation, we obtain
%==============================
\begin{align}
    N^{(n)}_k
    = 1 + \sum^n_{j=1} \sum^k_{i=1} b^{(j)}_i
    = 1 + (n-\lceil \log_2 k \rceil)k + \sum^{\lceil \log_2 k \rceil}_{j=1} 2^{j-1}
    = (n-r)k + 2^r,
\end{align}
%==============================
where $r:=\lceil \log_2 k \rceil$ (see Fig.~\ref{fig:bmatrix}).

When both real and imaginary parts of $\big<\varphi\big|b\big>\big<c\big|\varphi\big>$ are required, the number of unique circuits is doubled.
Hence, the upper bound of the number of unique circuits is $2((n-r)k+2^r)$.
When evaluating $\langle \psi_0 | A' | \psi_1 \rangle=\langle \varphi | A | \varphi \rangle$ with $|\psi_0 \rangle=|\psi_1\rangle$, the matrix $A$ can have non-zero diagonal components.
Considering that ${\rm Im}\big(\big<\varphi\big|b\big>\big<b\big|\varphi\big>\big)={\rm Im}\big(|\big<b\big|\varphi\big>|^2\big)=0$,
the upper bound of the number of unique circuits can be reduced to $2((n-r)k+2^r)-1$ in this case.

\qed

%==================================================
\subsection{Na\"ive Pauli measurements}
\label{app:proofs3}
%==================================================
When $A_{bc}\neq 0$, ${\rm Tr}\big(A^\dagger{\bf P}\big)$ takes a non-zero value where ${\bf P}:=P_{(b\oplus c)_{n-1}} \otimes P_{(b\oplus c)_{n-2}} \otimes \cdots \otimes P_{(b\oplus c)_{0}}$, $P_0\in\{I,Z\}$, $P_1\in\{X,Y\}$, and $I$, $X$, $Y$ and $Z$ are the Pauli matrices.
Therefore, for each distinct $b\oplus c$ such that $A_{bc}\neq0$, $2^n$ Pauli strings are required to express $\big<\psi_0\big|A\big|\psi_0\big>=\sum_{\bf P}\Big({\rm Tr}\big(A^\dagger{\bf P}\big)/2^n\Big)\big<\psi_0\big|{\bf P}\big|\psi_0\big>$.
The number of distinct $b\oplus c$ is the same as the number of colors in the colored matrix defined in Sec.~\ref{app:colored}, and the number of unique Pauli strings required to express $\big<\psi_0\big|A\big|\psi_0\big>$ is therefore $2^n((n-r)k+2^r)$.
In the case that $\big<\psi_0\big|A\big|\psi_1\big>$, ${\bf P}:=P_n \otimes P_{(b\oplus c)_{n-1}} \otimes P_{(b\oplus c)_{n-2}} \otimes \cdots \otimes P_{(b\oplus c)_{0}}$ is used instead where $P_n\in\{X,Y\}$ and the number of unique Pauli strings required to express $\big<\psi_0\big|A\big|\psi_1\big>$ is therefore $2^{n+1}((n-r)k+2^r)$.

%==================================================
\subsection{Evaluating Pauli strings by XBM}
\label{app:xbmpauli}
%==================================================
Consider evaluating an expectation value of one given Pauli string using the XBM method.
Let $\sigma_{n-1} \sigma_{n-2} \cdots \sigma_1 \sigma_0$ be a Pauli string where $\sigma_i \in \{I,X,Y,Z\}$.
Converting this Pauli string into a row-column basis gives a matrix $A$ where $A_{b, b\oplus p}\neq 0$ for all $b=0,1,\cdots,2^n-1$ and $p$ is a bit string defined as
%==============================
\begin{align}
\label{eq:}
    p_i
    &:=\left\{
        \begin{array}{ll}
            0 & {\rm if}\ \ \ \sigma_i \in \{I,Z\}\\[5mm]
            1 & {\rm if}\ \ \ \sigma_i \in \{X,Y\}
        \end{array}
    \right. .
\end{align}
%==============================
The non-zero elements of $A$ are simultaneously measurable using the XBM method because the XBM method groups all $A_{bc}$ that satisfy $b\oplus c=b'\oplus c'$, and $b\oplus (b\oplus p)=p$.
For example, consider $XIYYZ$.
In this case, $p=10110$ and non-zero elements of $A$ are $A_{00000,10110}, A_{00001,10111},\cdots,A_{11111,01001}$.
All components of such $A$ are simultaneously measurable using the XBM method because $00000\oplus10110=00001\oplus 10111=\cdots=11111\oplus01001=10110=p$.
Therefore, the number of unique circuits to evaluate $\big<\varphi\big|XIYYZ\big|\varphi\big>$ using the XBM method is two (corresponding to the real and imaginary parts).

\end{document}